\documentclass[12pt]{article}
\usepackage{graphicx}
\usepackage{wrapfig}
\usepackage{amsthm,color}
\makeatletter

\@addtoreset{equation}{section}
\makeatother
\usepackage{amsfonts}
\usepackage{eucal}
\usepackage{amsmath}
\usepackage{amssymb}
\usepackage{bm}
\usepackage{bbm}

\def\hsymbu#1{\smash{\lower1.7ex\hbox{\huge$#1$}}}
\newtheorem{theorem}{\bf Theorem}[section]
{\bf Definition}

\newtheorem{lemma}[theorem]{\bf Lemma}
\newtheorem{proposition}[theorem]{\bf Proposition}
\newtheorem{remark}
{\bf Remark}[section]
\newtheorem{example}
{\bf Example}[section]
\newtheorem{assumption}{\bf Assumption}

\newcommand{\Ran}{{\rm ran}}

\begin{document}
\title{\bf 
	Localization of a multi-dimensional \\
	quantum walk with one defect
}
\author{Toru Fuda
	\footnote{
		Department of Mathematics, Hokkaido University, 
		Sapporo 060-0810, Japan, \newline 
		E-mail: t-fuda@math.sci.hokudai.ac.jp
	}, 
	Daiju Funakawa
	\footnote{
		Department of Mathematics, Hokkaido University, 
		Sapporo 060-0810, Japan, \newline
		E-mail: funakawa@math.sci.hokudai.ac.jp
	}
	and 
	Akito Suzuki
	\footnote{
		Division of Mathematics and Physics, Faculty of Engineering, 
		Shinshu University, Wakasato, Nagano 380-8553,  Japan, E-mail:
		akito@shinshu-u.ac.jp
	}
}
\date{}

\maketitle

\begin{abstract}
In this paper,
we introduce a multidimensional generalization of 
Kitagawa's split-step discrete-time quantum walk,
study the spectrum of its evolution operator 
for the case of one defect coins,
and prove localization of the walk. 
Using a spectral mapping theorem,
we can reduce the spectral analysis of the evolution operator 
to that of a discrete Schr\"{o}dinger operator with variable coefficients,
which is analyzed using the Feshbach map. 
\end{abstract}

\medskip

\begin{flushleft}
%

\end{flushleft}

\section{Introduction}
Quantum walks (QWs) have been introduced
and studied in various contexts 
such as quantum probability \cite{Gud}, quantum optics \cite{Ah}, 
quantum cellular automata, \cite{Gro, Mey}, 
and quantum information \cite{Am01, CFG}
(see \cite{Am03, Kem03, Ko08, VA15} for more details). 
Among them, 
motivated by Grover's quantum search algorithms\cite{Gr, SKW,FGB}, 
researchers have proposed several types of discrete time QWs on graphs 
\cite{Wa, AhAm, Ken, IKK04, MBSS, TFMK}.  
Szegedy \cite{Sz} introduced a bipartite walk, 
which is defined on a bipartite graph, 
to construct a quantum search algorithm. 
Magniez et al \cite{MNRoS07, MNRiS07} 
updated the notion of bipartite walks 
and Segawa \cite{Se13} redefined an evolution operator $U_G$
on the Hilbert space $\ell^2(D)$ of square summable functions on 
the set $D$ of arcs for a digraphs $G= (V, D)$. 
The QW defined by $U_G$ 
is now referred to as the Szgedy walk on $G$,
which includes the Grover walk on $G$ as a special case. 
The Szegedy walks have a spectral mapping property 
from the transition probability matrix
$P_G$ 
of a random walk on $G$
to the evolution $U_G$,
which gives a useful tool for analyzing the spectrum of $U_G$
(see \cite{Se13, MS15, HS15} for more details). 
An extended version of the Szegedy walk, the twisted Szegedy walk,
was introduced by Higuchi et al \cite{HKSS14}
to study the spectral and asymptotic properties of the Grover walks 
on crystal lattices.
Higuchi, Segawa, and one of the authors of this paper \cite{HSS16}
proved the spectral mapping theorem (SMT)
for more general 
evolution $U = SC$,
where $S$ and $C$ are unitary and self-adjoint 
on a Hilbert space $\mathcal{H}$  
and where $C$ is assumed to be of the form 
\[ C = 2d^* d -1 \] 
with a coisometry  $d$ from $\mathcal{H}$ to 
a Hilbert space $\mathcal{K}$,
i.e., $dd^*$ is the identity $I_{\mathcal{K}}$ on $\mathcal{K}$
. 
Observe that the Hilbert space $\mathcal{H}$ here can be taken
to be arbitrary and is no longer needed to be $\ell^2(D)$. 
Let $T= d S d^*$. 
$T$ is a self-adjoint operator on $\mathcal{K}$
and called the discriminant operator of $U$. 
Let $\mathcal{D}_\pm =\ker d \ \cap \ker (S \pm 1)$.
The subspace 
$\mathcal{D}_{\rm B} 
= \mathcal{D}_+ \oplus \mathcal{D}_- \subset \mathcal{H} $
is called the {\it birth eigenspace} of $U$ and its orthogonal complement 
$\mathcal{D}_{\rm I}$ the {\it  inherited subspace} of $U$
(see \cite{HS15,MOS16}). 
As shown elsewhere \cite{SS16}, 
the restriction $U_{\rm I} :=U|_{\mathcal{D}_{\rm I}}$ 
to the inherited subspace is unitarily equivalent to 
\[ \exp(+i\arccos T)\oplus \exp(-i\arccos T) \]
and the restriction $U_{\rm B} :=U|_{\mathcal{D}_{\rm B}}$
to the birth eigenspace is 
$I_{\mathcal{D}_+} \oplus (-I_{\mathcal{D}_-})$. 
Thus, the spectral analysis of $U$ is reduced to two parts:
(1) the spectral analysis of $T$ and
(2) the calculation of $\dim \mathcal{D}_{\rm B}$.  
This reduction leads the SMT from $T$ to $U$
(Theorem \ref{thm_SM}), 
which allows us to use it
for QWs other than the Szegedy walk.
As evident below, 
such an abstract theorem is applicable 
for a class of $d$-dimensional QWs,
which is not the Szegedy walk on $\mathbb{Z}^d$.  
In forthcoming papers \cite{FFSwlt, FFSloc}, 
we will consider a unified model
that includes a split-step QW 
introduced by Kitagawa et al \cite{KRBD}
and traditional one-dimensional QWs 
\cite{Am01,Gud,Mey} as special cases. 
The evolution of the walk
is a unitary operator on 
$\ell^2(\mathbb{Z};\mathbb{C}^2)$ defined as
$U = S_1 C$, 
where 
\[ (S_1\psi)(x) = \begin{pmatrix} 
	p \psi_1(x) + q \psi_2(x+1) \\ 
	q^* \psi_1(x-1) - p \psi_2(x) \end{pmatrix},
		\quad x \in \mathbb{Z},
		\; \psi \in \ell^2(\mathbb{Z};\mathbb{C}^2). 
\]
Taking $(p,q) \in \mathbb{R} \times \mathbb{C}$
as $p^2 + |q|^2=1$ ensures $S_1$ is unitary and self-adjoint.
$C$ is a multiplication operator by unitary matrices $C(x) \in U(2)$. 
If $C(x)$ is in addition hermitian and 
$\dim \ker(C(x)-1) = 1$ for all $x \in \mathbb{Z}$,
then $C$ is written as $2d^* d -1$
with a coisometry 
$d:\ell^2(\mathbb{Z};\mathbb{C}^2) \to \ell^2(\mathbb{Z})$
(see Example \ref{ex_SS}).
Thus the SMT is applicable. 

\paragraph{\em Models}
In this paper, 
we consider a multi-dimensional generalization of 
the aforementioned model,
which is a $2d$-state QW on $\mathbb{Z}^d$ 
with a position dependent coin $C({\bm x}) \in U(2d)$ and $d \geq 2$. 
However, for conceptual and notational simplicity,
we first concentrate on the case of $d=2$. 
The case of $d \geq 3$
is dealt with in the subsequent sections. 
Let $\mathcal{H} = \ell^2(\mathbb{Z}^2 ; \mathbb{C}^4)$
be the Hilbert space of states.  
As usual, the evolution operator $U=SC$ is defined as
the product of a shift $S$ and a coin $C$. 
To define the shift operator, we introduce a set
\[ D = \{ ({\bm p},{\bm q}) = (p_1,p_2,q_1,q_2) 
	\in \mathbb{R}^2 \times \mathbb{C}^2 
			\colon  \mbox{$p_j^2+|q_j|^2=1$ ($j=1,2$)} \} \]	
and use $\{\bm{e}_j\}_{j=1}^2$ 
to denote the standard basis of $\mathbb{Z}^2$.
We define operators $S_j$ ($j=1,2$) on 
$\ell^2(\mathbb{Z}^2 ; \mathbb{C}^2)$ as 
\[
	(S_j\Psi)(\bm{x})=
	\begin{pmatrix}
		p_j\psi_1(\bm{x})+q_j\psi_2(\bm{x}+\bm{e}_j) \\
		q_j^{\ast}\psi_1(\bm{x}-\bm{e}_j)-p_j\psi_2(\bm{x})
	\end{pmatrix},
	\quad
	\bm{x}\in \mathbb{Z}^2,
\]
for 	all $\Psi= {}^t (\psi_1, \psi_2)
	\in \ell^2(\mathbb{Z}^2 ; \mathbb{C}^2)$.
The shift $S$ on $\mathcal{H}$ is defined as
a diagonal operator 
$
	S= S_1 \oplus  S_2
$
on 
$\mathcal{H} \simeq
\oplus^2 \ell^2(\mathbb{Z}^2 ; \mathbb{C}^2)$
.
The condition $({\bm p},{\bm q}) \in D$ ensures 
that $S_j$ is self-adjoint and unitary on 
$\ell^2(\mathbb{Z}^2 ; \mathbb{C}^2)$
and so is $S$ on $\mathcal{H}$. 
The coin operator is a multiplication by
unitary and self-adjoint square matrices $C(\bm{x}) \in U(4)$. 
In general, a unitary and self-adjoint operator
is an involution; hence, it can only have eigenvalues $\pm 1$ 
as its spectrum. 
We 
impose the following on the coin operator $C$.  
\begin{itemize}
	\item ({\bf Simplicity})\quad 
$\dim \ker(C(\bm{x})-1)=1,
		\quad \bm{x}\in\mathbb{Z}^2$. 
	\item ({\bf One defect})\quad 
		$
		C(\bm{x})
		=
		\begin{cases}
			C_1, & {\bm x}\in \mathbb{Z}^2\setminus\{\bm{0}\} \\
			C_0, & {\bm x}=\bm{0}
		\end{cases}
	$ 
	with some $C_0, C_1 \in U(4)$.  
\end{itemize}

We here comment on the aforementioned conditions. 
The simplicity condition means that 
$C({\bm x})$ is a Grover-type coin. 
Indeed, by $\dim \ker(C(\bm{x})-1)=1$,
we can take a unique normalized eigenvector 
$
\chi(\bm{x})
\in \ker(C(\bm{x})-1)
$ up to a constant factor.
As seen in Lemma \ref{lem_d}, 
we can write  $C = 2d^* d- 1$
with a coisometry 
$d: \mathcal{H} \to \mathcal{K}:= \ell^2(\mathbb{Z}^2)$ 
defined as
\[ (d \Psi)({\bm x}) 
= \langle \chi({\bm x}), \Psi({\bm x}) \rangle_{\mathbb{C}^2},
\quad {\bm x} \in \mathbb{Z}^2
\quad \mbox{for all $\Psi \in \mathcal{H}$.}
\] 
The one defect condition means that
$\chi(\bm{x})$ can be written as 
\begin{align*}
		\chi(\bm{x})
		&=
		\begin{cases}
			\Phi =
				{}^t(\Phi_1,\Phi_2) \ 
			\text{with}\ 
			\Phi_j
			{\in\mathbb{C}^2}
			\ (j=1,2),
			& 
			\bm{x}\in\mathbb{Z}^2 \setminus\{\bm{0}\}, \\
			\Omega =
				{}^t (\Omega_1,\Omega_2)\
			\text{with}\
			\Omega_j 
			{\in\mathbb{C}^2}
			\ (j=1,2),
			& 
			\bm{x}=\bm{0},
		\end{cases}
\end{align*}
where $\Phi \in \ker (C_1-1)$ and $\Omega \in \ker (C_0-1)$
are normalized vectors. 
In Grover's search algorithm on a graph $G =(V,E)$,
the coin operator $C({\bm x})$ differs only 
at a vertex ${\bm x} = {\bm x}_0$,
which is a unique solution to the search problem
. 
This is a one-defect condition. 
Moreover, finding the marked vertex ${\bm x}_0$
with non-zero probability 
is closely related to localization of the corresponding QW. 
Motivated by Grover's search algorithm, 
we study localization of the one defect model on $\mathbb{Z}^d$. 


\paragraph{\em Results}
Let $\Psi_0 \in \mathcal{H}$ be the initial state of a quantum walker,
and let $\Psi_t = U^t \Psi_0$ ($t =1,2,\ldots$) be
the state of the walker at time $t$. 
The position $X_t$ of the walker at time $t$
follows 
$P(X_t={\bm x}) = \|\Psi_t({\bm x})\|^2_{\mathbb{C}^2}$
(${\bm x} \in \mathbb{Z}^2$). 
As shown in \cite{SS16}, 
if the initial state $\Psi_0$ has a overlap with an eigenvector of $U$,
then localization occurs, i.e.,
\[ \limsup_{t \to \infty} P(X_t={\bm x}) > 0
	\ \mbox{with some ${\bm x} \in \mathbb{Z}^2$}. 
\]
Thus the problem is reduced to proving 
the existence of eigenvalues for $U$. 

We are now in a position to state our result.  
Let
$\sigma_1 = \begin{pmatrix} 0 & 1 \\ 1 & 0 \end{pmatrix}$,
$\sigma_+ = \begin{pmatrix} 0 & 1 \\ 0 & 0 \end{pmatrix}$,
and $\sigma_3 = \begin{pmatrix} 1 & 0 \\ 0 & -1 \end{pmatrix}$. 
We set
$\displaystyle 
a_\Omega({\bm p})
	= \sum_{j=1}^2 p_j 
	\langle \Omega_j, \sigma_3 \Omega_j \rangle_{\mathbb{C}^2}$,
$\displaystyle 
a_\Phi({\bm p})
	= \sum_{j=1}^2 p_j 
	\langle \Phi_j, \sigma_3 \Phi_j \rangle_{\mathbb{C}^2}$,
$\displaystyle 
\lambda({\bm q}) 
	= 2\sum_{j=1}^2 \left|q_j \langle \Phi_j, \sigma_+ \Phi_j 
		\rangle_{\mathbb{C}^2}\right|$,   
and
$D_j = \{ ({\bm p}, {\bm q}) \in D \colon p_jq_j \not=0 \}$
($j=1,2$).  
Let 
$\mathbb{T}_- = [-1, -\lambda({\bm q})+a_{\Phi}({\bm p}))$,
$\mathbb{T}_+ = (\lambda({\bm q})+a_{\Phi}({\bm p}), 1]$,
and $g_\pm(\lambda)=e^{\pm i\arccos \lambda}$.  
{We use \, $\cdot$ \, to denote the scalar product. }

\begin{wrapfigure}{r}{75mm}
\unitlength 0.1in
\begin{picture}( 29.0500, 25.5000)(  2.6000,-26.2000)
%
\special{pn 4}%
\special{ar 1836 1400 1000 1000  0.0000000 6.2831853}%
%
\special{pn 8}%
\special{pa 356 1400}%
\special{pa 3166 1400}%
\special{fp}%
\special{sh 1}%
\special{pa 3166 1400}%
\special{pa 3098 1380}%
\special{pa 3112 1400}%
\special{pa 3098 1420}%
\special{pa 3166 1400}%
\special{fp}%
%
\special{pn 8}%
\special{pa 1836 2620}%
\special{pa 1836 70}%
\special{fp}%
\special{sh 1}%
\special{pa 1836 70}%
\special{pa 1816 138}%
\special{pa 1836 124}%
\special{pa 1856 138}%
\special{pa 1836 70}%
\special{fp}%
%
\special{pn 20}%
\special{ar 1836 1400 1000 1000  4.3128735 5.8151518}%
%
\special{pn 20}%
\special{ar 1836 1410 1000 1000  0.4680335 1.9703118}%
%
\special{pn 8}%
\special{pa 1426 490}%
\special{pa 1426 2330}%
\special{dt 0.045}%
%
\special{pn 8}%
\special{pa 2736 970}%
\special{pa 2736 1840}%
\special{dt 0.045}%
%
\special{pn 8}%
\special{pa 2086 1470}%
\special{pa 2086 1330}%
\special{fp}%
%
\special{pn 20}%
\special{sh 1}%
\special{ar 1016 1970 10 10 0  6.28318530717959E+0000}%
%
\special{pn 8}%
\special{pa 1016 840}%
\special{pa 1016 1970}%
\special{dt 0.045}%
\put(29.2500,-14.9000){\makebox(0,0)[lt]{1}}%
\put(5.8500,-14.9000){\makebox(0,0)[lt]{$-1$}}%
\put(20.9000,-2.8000){\makebox(0,0){{\footnotesize $\sigma(U_{\rm I})$}}}%
\put(20.8500,-16.1000){\makebox(0,0){{\footnotesize $a_{\Phi}({\bm p})$}}}%
\put(9.1000,-12.9000){\makebox(0,0)[lb]{{\footnotesize $-\lambda({\bm q})+a_{\Phi}({\bm p})$}}}%
\put(27.4000,-12.3000){\makebox(0,0){{\footnotesize $\lambda({\bm q})+a_{\Phi}({\bm p})$}}}%
%
\special{pn 8}%
\special{pa 2736 1490}%
\special{pa 2736 1350}%
\special{fp}%
%
\special{pn 8}%
\special{pa 1426 1490}%
\special{pa 1426 1350}%
\special{fp}%
\put(2.6000,-4.4000){\makebox(0,0)[lt]{{\footnotesize Eigenvalue of $U_{\rm I}$}}}%
\put(2.6000,-22.1000){\makebox(0,0)[lt]{{\footnotesize Eigenvalue of $U_{\rm I}$}}}%
%
\special{pn 4}%
\special{pa 940 610}%
\special{pa 1000 770}%
\special{fp}%
\special{sh 1}%
\special{pa 1000 770}%
\special{pa 996 702}%
\special{pa 982 720}%
\special{pa 958 716}%
\special{pa 1000 770}%
\special{fp}%
%
\special{pn 4}%
\special{pa 900 2160}%
\special{pa 990 2020}%
\special{fp}%
\special{sh 1}%
\special{pa 990 2020}%
\special{pa 938 2066}%
\special{pa 962 2066}%
\special{pa 972 2088}%
\special{pa 990 2020}%
\special{fp}%
%
\special{pn 13}%
\special{pa 1430 1410}%
\special{pa 2740 1410}%
\special{fp}%
%
\special{pn 20}%
\special{sh 1}%
\special{ar 1020 1400 10 10 0  6.28318530717959E+0000}%
\special{sh 1}%
\special{ar 1020 1400 10 10 0  6.28318530717959E+0000}%
\put(20.9000,-11.3000){\makebox(0,0){{\footnotesize $\sigma(T)$}}}%
%
\special{pn 8}%
\special{pa 2090 940}%
\special{pa 2090 630}%
\special{fp}%
\special{sh 1}%
\special{pa 2090 630}%
\special{pa 2070 698}%
\special{pa 2090 684}%
\special{pa 2110 698}%
\special{pa 2090 630}%
\special{fp}%
%
\special{pn 8}%
\special{pa 2090 1840}%
\special{pa 2090 2150}%
\special{fp}%
\special{sh 1}%
\special{pa 2090 2150}%
\special{pa 2110 2084}%
\special{pa 2090 2098}%
\special{pa 2070 2084}%
\special{pa 2090 2150}%
\special{fp}%
\put(21.5000,-7.9000){\makebox(0,0)[lb]{{\footnotesize $g_+$}}}%
\put(21.6000,-19.8000){\makebox(0,0)[lb]{{\footnotesize $g_-$}}}%
%
\special{pn 20}%
\special{sh 1}%
\special{ar 1020 830 10 10 0  6.28318530717959E+0000}%
\end{picture}%
  \vspace{-3mm}
  \caption{Location of the spectrum $\sigma(U_{\rm I})$
  for $a_\Omega({\bm p}_0) < a_{\Phi}({\bm p}_0)$. 
  $g_\pm (\lambda) = e^{\pm i\arccos \lambda}$ 
  map
  $\sigma(T) = [-\lambda({\bm q}) + a_{\Phi}({\bm p}),
  \lambda({\bm q}) + a_{\Phi}({\bm p})]
  $  onto $\sigma(U_{\rm I}) \subset S^1$.
  The difference $\sigma(U) \setminus \sigma(U_{\rm I})$
  is at most $\sigma(U_{\rm B}) \subset \{-1,+1\}$. 
  See Theorem \ref{MThm} for more details.}
  \label{fig01}
\end{wrapfigure}
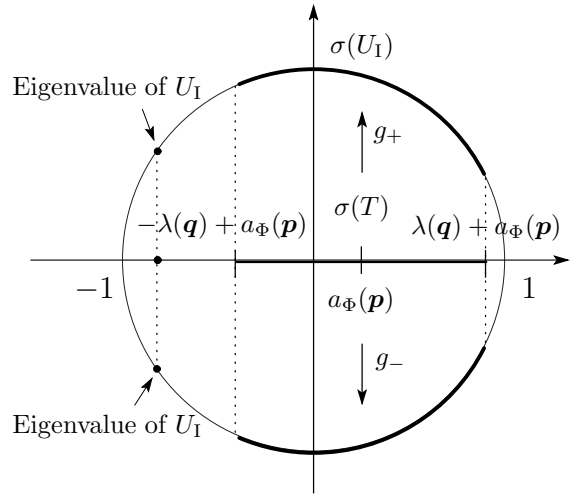

\begin{theorem}
\label{thm_d=2}
Let $U=SC$ as above. Suppose that the following conditions hold. 
\begin{itemize}
\item[(1)] 
{$\Phi_j \cdot (\sigma_1 \Omega_j) = 0$}
for all $j \in \{1,2\}$
and
{$\langle \Phi_l, \sigma_+ \Omega_l \rangle_{\mathbb{C}^2} 
\not= 0$} with some $l \in \{1,2\}$;
\item[(2)]
$a_\Omega({\bm p}_0) \not=a_\Phi({\bm p}_0)$
with some ${\bm p}_0 \in \{-1,1\} \times \{-1,1\}$. 
\end{itemize}
If 
{$({\bm p}, {\bm q}) \in D_l$} and
$\|({\bm p}, {\bm q}) - ({\bm p}_0, {\bm 0})\|_{\mathbb{R}^2\times \mathbb{C}^2}$
is sufficiently small , $U$ has two eigenvalues in 
$\{ g_-(\lambda), g_+(\lambda) \mid 
	\lambda \in \mathbb{T}_- \}$
or
$\{ g_-(\lambda), g_+(\lambda) \mid 
	\lambda \in \mathbb{T}_+\}$. 
\end{theorem}
This is a special case of Theorem \ref{MThm}. 
See Figure \ref{fig01} for 
the location of the eigenvalues and the continuous spectrum.  
The criteria for $U_{\rm I}$ to have eigenvalues 
in 
{$\{ g_-(\lambda), g_+(\lambda) \mid 
	\lambda \in \mathbb{T}_- \}$
and $\{ g_-(\lambda), g_+(\lambda) \mid 
	\lambda \in \mathbb{T}_+ \}$
}%
are obtained in Theorem \ref{mainth}.

\paragraph{\em Methods and related work}
Localization of the one defect model of 
traditional one-dimensional QWs was solved
by Cantero et al \cite{CGMV},
who used  the CGMV method,
which is not applicable for multidimensional cases. 
In the present work, we use the SMT. 
Several studies on the birth eigenspace 
$\mathcal{D}_{\rm B}$ have been reported.   
As shown by Higuchi et al \cite{HKSS14}, 
multi-dimensional models
are likely to have eigenvalues $\pm 1$
due to the existence of cycles,
which makes $\mathcal{D}_{\rm B}$ non trivial. 
For the one-dimensional split-step QW,
the birth eigenspace is characterized elsewhere \cite{FFSloc}. 
However, 
the eigenspace contained in the inherited subspace 
$\mathcal{D}_{\rm I}$
is study only for the Grover walk case,
in which the discriminant $T$ is unitarily equivalent to 
the transition probability matrix $P_G$
of the symmetric random walk on $G$. 
In our case, 
$T$ becomes a discrete Schr\"{o}dinger operator 
with variable coefficients: 
	\begin{align*}
		T=
		a(\bm{p}, \cdot)+
		\sum_{j=1}^2
		\left\{
		q_j \langle \chi_j,  L_j \sigma_+ \chi_j \rangle
		+ (q_j \langle \chi_j,  L_j \sigma_+ \chi_j \rangle)^*
		\right\},
	\end{align*}
where 
$\chi({\bm x}) 
	= 
	{}^t(\chi_1({\bm x}), \chi_2({\bm x}) )
$,
$a(\bm{p}, \bm{x})=\sum_{j=1}^2 p_j
	\langle \chi_j(\bm{x}), \sigma_3 \chi_j(\bm{x}) \rangle$,
and $L_j$ is the shift by ${\bm e}_j$ on $\mathcal{K}$. 
To analyze the above operator $T$,
we employ the Feshbach map \cite{F}: 
$$
	F(\lambda) = 
	\Pi^{\perp}(T-\lambda)\Pi^{\perp}-\Pi^{\perp}T\Pi 
	(\Pi (T-\lambda)\Pi)_{{\Ran}\Pi}^{-1} 
	\Pi T \Pi^{\perp},
	\quad \lambda \in \mathbb{C} 
		\setminus \{a_{\Omega}({\bm p})\},
$$
where $\Pi$ is the projection onto 
$\{ \psi \in \mathcal{K} \mid \psi({\bm x}) = 0 \;
\mbox{except for ${\bm x}=0$}  \}$
and $\lambda$ is a spectral parameter. 
The isospectral property of this map implies that
$\lambda$ is an eigenvalue of $T$
if 
$\ker F(\lambda)$ is non-trivial 
(Proposition \ref{feshbach}). 
The Feshbach map was used in a study of nuclear reactions \cite{F}
and was used for constructing a 
renormalization map
\cite{BFS}. 
To our best  knowledge, 
this is the first application of the Feshbach map 
to analyze the spectrum of an evolution operator for a QW. 
The one defect condition yields the following formula:
\begin{equation}
\label{FSM}
		F(\lambda) = \Pi^{\perp}\left(
		T_0-\lambda -\frac{1}{a_{\Omega}(\bm{p})-\lambda}|\varphi_{\bm{q}} \rangle 
	\langle \varphi_{\bm{q}} |
	\right)
	\Pi^{\perp},
\end{equation}
where $\varphi_{\bm q} \in \mathcal{K}$. 
This is a one rank perturbation of 
a constant coefficient discrete Laplacian $T_0$. 
The spectral analysis of an operator 
$(1/\sqrt{d}) \sum_{j=1}^d (L_j + L_j^*) + v |\delta_{\bm 0}\rangle \langle \delta_{\bm 0}|$
(with $v$ a coupling constant and
$\delta_{\bm 0}$ the delta function at the origin)  
similar to the right-hand side in \eqref{FSM} is treated 
elsewhere \cite{HSSS}. 
Because the nonlinearity of the spectral parameter $\lambda$,
the analysis of the kernel of $F(\lambda)$ becomes more involved.
This task is reduced to finding zeros of a function 
\[ 
	\mathfrak{f}(\lambda)
	=
	\lambda -a_{\Omega}(\bm{p})+
	\langle \varphi_{\bm{q}}, 
		(T_0-\lambda)^{-1}  \varphi_{\bm q} \rangle_{\mathcal{K}},
	\quad \lambda\in [-1,1]\setminus\sigma(T_0)\neq \emptyset. 
\]

The rest of this paper is constructed as follows. 
In Sec. 2, we review the SMT, 
which plays an important role in this work. 
The precise definitions of our evolution $U$
and the discriminant $T$ are given in  Sec. 3.
We  thereafter give the essential spectrum of $T$, 
which is mapped onto the essential spectrum of $U_{\rm I}$ by the SMT. 
We also give a criterion for $T$ 
to have an eigenvalue in terms of the Feshbach map
$F(\lambda)$ (Theorem \ref{f0}). 
We  then present the main results.  
Theorem \ref{mainth} gives criteria for $U$ to have eigenvalues and 
Theorem \ref{MThm} shows the existence of eigenvalues for $U$. 
We prove Theorem \ref{MThm} using Theorem \ref{mainth}. 
Sec. 4 is devoted to the precise definition of the Feshbach map
$F(\lambda)$
and its properties. 
Sec. 5 is devoted to the analysis of $\mathfrak{f}(\lambda)$ 
and the proof of  
Theorem \ref{mainth}.

\section{Preliminaries}
In this section, we briefly review the spectral mapping theorem (SMT). 
Readers can consult \cite{HSS16,SS16} for more details. 
We use $\sigma(A)$, $\sigma_{\rm p}(A)$, 
$\sigma_{\rm ac}(A)$, $\sigma_{\rm sc}(A)$
to denote the spectrum, the set of all eigenvalues,
the absolutely continuous spectrum, 
and the singular continuous spectrum
of an operator $A$, respectively. 
Let $\mathcal{H}$ and $\mathcal{K}$ be Hilbert spaces
and $d:\mathcal{H} \to \mathcal{K}$ be a coisometry,
i.e., 
\[ dd^* = I_{\mathcal{K}}, \]
where $d^*:\mathcal{K} \to \mathcal{H}$ is the adjoint of $d$ 
and $I_{\mathcal{K}}$ is the identity  on $\mathcal{K}$. 
Then, 
$d$ is an isometry and 
$d^* d$ is a othogonal projection on $\mathcal{H}$, 
because $d^*d$ is idempotent and self-adjoint, 
i.e., $(d ^*d)^2 = d^* d$ and $(d^* d)^* = d^*d$. 
The operator
\[ C := 2d^* d -1  \] 
is a self-adjoint unitary operator,
because $C^2=1$.  
Let $S$ be a self-adjoint unitary operator on $\mathcal{H}$ and
set $U = SC$. 
The discriminant operator $T$ of $U$ is defined as 
\[ T= d S d^*, \] 
which is a bounded self-adjoint operator on $\mathcal{K}$
and $\|T\| \leq 1$. 
Hence, $\sigma(T) $ is a closed set contained in the interval 
$[-1,1]$. 
Let $\mathcal{D}_\pm =\ker d \ \cap \ker (S \pm 1) 
\subset \mathcal{H}$.
The subspaces 
\[ \mathcal{D}_{\rm B} 
:= \mathcal{D}_+ \oplus \mathcal{D}_- 
\quad \mbox{and} \quad
\mathcal{D}_{\rm I} := \mathcal{D}_{\rm B}^\perp \] 
are called the {\it birth eigenspace} of $U$ 
and {\it  inherited subspace} of $U$, respectively.
The restriction $U_{\rm I} :=U|_{\mathcal{D}_{\rm I}}$ 
to the inherited subspace is unitarily equivalent to 
\[ \exp(+i\arccos T)\oplus \exp(-i\arccos T) 
	\quad \mbox{on ${\Ran} (d^*d)$}. \]
See \cite{SS16} for the precise meaning of the above decomposition. 
On the other hand, 
the restriction $U_{\rm B} :=U|_{\mathcal{D}_{\rm B}}$
to the birth eigenspace coincides with 
$I_{\mathcal{D}_+} \oplus (-I_{\mathcal{D}_-})$. 
The SMT from $T$ to $U$ is given as follows. 
\begin{theorem}[Spectral mapping theorem \cite{HSS16,SS16}]
\label{thm_SM}
Let $U=SC$ be as above. 
Then, $U$ is decomposed into
$U = U_{\rm I} \oplus U_{\rm B}$
on $\mathcal{H} 
		= \mathcal{D}_{\rm I} \oplus \mathcal{D}_{\rm B}$
and the following hold:
\begin{enumerate}
\item $\sigma_\sharp(U) = \sigma_\sharp(U_{\rm I})$
	for $\sharp = {\rm ac}, {\rm sc}$
and $\sigma_{\rm p}(U) 
	= \sigma_{\rm p}(U_{\rm I}) \cup \sigma_{\rm p}(U_{\rm B})$; 
\item $ \sigma_\sharp(U_{\rm I}) 
	= \exp({+ i \arccos \sigma_\sharp(T)} )
		\cup  \exp({- i \arccos \sigma_\sharp(T)})$
		for $\sharp = {\rm p}, {\rm ac}, {\rm sc}$;
\item $\sigma_\sharp(U_{\rm B}) = \emptyset$ 
		for $\sharp = {\rm ac}, {\rm sc}$ and
$\sigma_{\rm p}(U_{\rm B})
		= \begin{cases} 
			\{1, -1\} & \mbox{if $\mathcal{D}_+ \not=\emptyset$
				and $\mathcal{D}_- \not=\emptyset$}, \\
			\{\pm1\} & \mbox{if $\mathcal{D}_\pm \not=\emptyset$
				and $\mathcal{D}_\mp =\emptyset$}, \\
			\emptyset & \mbox{otherwise}.
		\end{cases}
$
\end{enumerate}
\end{theorem}

Theorem \ref{thm_SM} is widely applicable for the evolutions
of quantum walks. 
Here we give two examples. The first one is the Szegedy walk. 
See \cite{HKSS14} for the twisted Szegedy walk.   
\begin{example}[\cite{HKSS14, HSS16}]
Let $D$ be the set of arcs of a  symmetric digraph $G=(V,D)$
(possibly not bipartite)
and $\mathcal{H} = \ell^2(D)$ 
be the Hilbert space of square summable functions 
$\psi:D \to \mathbb{C}$. 
Define  a unitary operator $U_G$ on $\ell^2(D)$ 
as the product 
\[ U_G =S_{\rm f}C_\chi \] 
of a shift $S_{\rm f}$ and coin $C_{\chi}$. 
The shift $S_{\rm f}$ is defined as 
$(S_{\rm f}\psi)(e) = \psi(\bar{e})$ for $e \in D$,
where $\bar{e}$ stands for the inverse arc of $e$. 
The coin $C_\chi$ is defined as
\[ C_\chi  = \bigoplus_{x \in V} 
	\left(2|\chi(x)\rangle \langle \chi(x)| -1\right), \]
where we have used an identification 
$\ell^2(D) \simeq \bigoplus_{x \in V} \mathcal{H}_x$
with 
$\mathcal{H}_x 
= \overline{\rm Span}\{\psi \colon \psi(e) = 0, o(e) \not=x \}$
and $\chi(x) = \sum\limits_{e \in D; o(e)=x} \sqrt{p_{t(e),x}}\delta_e
\in \ell^2(D)$ is a normalized vector. 
Here $p_{u,v}$ is the transition probability of a (classical) random walk
from $v$  to $u$ ($u,v \in V$). 
The QW with this evolution $U_G$ is now referred
to as the Szegedy walk on $G$,
which is called the Grover walk on $G$ in particular 
if  $p_{u,v} = 1/{\rm deg} x$. 
In the case of the Szegedy evolution operator $U_G$,
Theorem \ref{thm_SM} is applicable 
for any symmetric digraph $G = (V,D)$,
because $S_{\rm f}$ is self-adjoint and unitary 
and $C_\chi = 2d_\chi^* d_\chi -1$
with a coisometry $d_\chi:\ell^2(D) \to \ell^2(V)$
defined as 
\[ (d_\chi \psi)(x) := \langle \chi(x), \psi \rangle,
	\quad x \in V. \] 
Moreover, the discriminant operator $T_G$ of $U_G$
is unitary equivalent to the transition probability matrix 
$P_G = (p_{u,v})$ and 
the birth eigenspace can be characterized by the structure of $G$. 
\end{example}
The next example is a one-dimensional QW 
but not the Szegedy walk on $\mathbb{Z}$.  
This is a unified model including a split-step QW 
introduced by Kitagawa et al \cite{KRBD}
and traditional one-dimensional QWs 
\cite{Am01,Gud,Mey} as special cases. 
In the subsequent sections, 
we consider a multidimensional extension of this model. 
\begin{example}[Split-step QWs]\label{ex_SS}
The evolution of a split-step QW
is a unitary operator on 
$\ell^2(\mathbb{Z};\mathbb{C}^2)$ defined as
$U = S_1 C$, 
where 
\[ (S_1\psi)(x) = \begin{pmatrix} 
	p \psi_1(x) + q \psi_2(x+1) \\ 
	q^* \psi_1(x-1) - p_2 \psi(x) \end{pmatrix},
		\quad x \in \mathbb{Z}. 
\]
We suppose that $(p,q) \in \mathbb{R} \times \mathbb{C}$
satisfy $p^2 + |q|^2=1$,
which ensure $S_1$ is unitary and self-adjoint.
$C$ is a multiplication operator by unitary matrices $C(x) \in U(2)$. 
When $p=0$ and $q=1$, it becomes a QW on 
$\mathbb{Z}$ with a flip-flop shift \cite{AKR},
which is unitarily equivalent to traditional 
QWs 
(see \cite{O} for more information). 
The evolutions with $p=0$ and $p\not=0$ 
are not unitarily equivalent and these walks have 
weak limit measures different from usual one \cite{FFSwlt}. 
If $C(x)$ is self-adjoint unitary and 
$\dim \ker(C(x)-1) = 1$ for all $x \in \mathbb{Z}$,
then $C$ is written as $2d_\chi^* d_\chi -1$
with a coisometry 
$d_\chi:\ell^2(\mathbb{Z};\mathbb{C}^2) \to \ell^2(\mathbb{Z})$
defined as 
\[ (d_\chi \Psi)(x) = \langle \chi(x), \Psi(x) \rangle_{\mathbb{C}^2},
	\quad x \in \mathbb{Z}, \]
where $\chi(x) \in \ker(C(x)-1)$. 
Thus the SMT is applicable for this model. 
In \cite{FFSloc}, the birth eigenspace of this model is characterized.  
\end{example}

\section{Multi-dimensional models and main results}
\subsection{Definition of models}
From now on, we consider a QW on $\mathbb{Z}^d$,
which is a generalization of the split-step QW defined 
in Example \ref{ex_SS}.  
Let $n\in\mathbb{N}$ and 
use $\ell^2(\mathbb{Z}^d ; \mathbb{C}^n)$ 
to denote the Hilbert space of 
the square-summable functions 
$\Psi : \mathbb{Z}^d\to \mathbb{C}^{n}$. 
If $n=1$, we simply denote $\ell^2(\mathbb{Z}^d ; \mathbb{C})$ by 
$\ell^2(\mathbb{Z}^d)$. 
Hereafter, we set 
$\mathcal{H}=\ell^2(\mathbb{Z}^d;\mathbb{C}^{2d})$ and 
$\mathcal{K}=\ell^2(\mathbb{Z}^d)$.
We fist define an evolution operator $U$ on $\mathcal{H}$
as a product $U=SC$ of a shift operator $S$ and coin operator $C$,
then introduce a coisometry $d:\mathcal{H} \to \mathcal{K}$,
and give an explicit formula of the discriminant operator 
$T=d S d^*$ on $\mathcal{K}$. 

\paragraph{Shift operators}
Let
\[ D = \{ ({\bm p},{\bm q}) = (p_1,\ldots, p_d,q_1,\ldots,q_d) 
	\in \mathbb{R}^d \times \mathbb{C}^d 
			\colon  \mbox{$p_j^2+|q_j|^2=1$ ($j=1,\ldots,d$)} \} \]	
and use $\{\bm{e}_j\}_{j=1}^d$ 
to denote the standard basis of $\mathbb{Z}^d$.
Henceforth $({\bm p}, {\bm q}) \in D$ is assumed 
unless otherwise specified. 
To define a shift operator $S$ on $\mathcal{H}$, we introduce
an operator $S_j$ on  $\ell^2(\mathbb{Z}^d ; \mathbb{C}^2)$ 
($j=1,\ldots,d$) as follows.
$$
	(S_j\psi)(\bm{x})=
	\begin{pmatrix}
		p_j\psi_1(\bm{x})+q_j\psi_2(\bm{x}+\bm{e}_j) \\
		q_j^{\ast}\psi_1(\bm{x}-\bm{e}_j)-p_j\psi_2(\bm{x})
	\end{pmatrix}
	\quad \text{for all}
	\
	\bm{x}\in \mathbb{Z}^d
	\
	\text{and}
	\
	\psi= 
	\begin{pmatrix}
		\psi_1\\
		\psi_2
	\end{pmatrix}
	\in \ell^2(\mathbb{Z}^d ; \mathbb{C}^2).
$$
Using the identification 
$\mathcal{H} \simeq \oplus_{j=1}^d 
\ell^2(\mathbb{Z}^d ; \mathbb{C}^2)$,
we define the shift $S$ on $\mathcal{H}$ as
$
	S= S_1 \oplus \ldots \oplus S_d
$, i.e., 
\begin{align*}
	(S\Psi)(\bm{x})=
	\begin{pmatrix}
		(S_1\Psi_1)(\bm{x})\\
		\vdots \\
		(S_d\Psi_d)(\bm{x})
	\end{pmatrix}
	\
	\text{for all}\ 
	\bm{x}\in \mathbb{Z}^d
	\
	\text{and}
	\
	\Psi=
	\begin{pmatrix}
		\Psi_1\\
		\vdots \\
		\Psi_d
	\end{pmatrix}
	\in \mathcal{H}\ 
		(\Psi_j\in\ell^2(\mathbb{Z}^d ; \mathbb{C}^2)).
\end{align*}
The condition $({\bm p},{\bm q}) \in D$ ensures 
$S_j$ is self-adjoint and unitary on 
$\ell^2(\mathbb{Z}^d ; \mathbb{C}^2)$,
and so is $S$ on $\mathcal{H}$. 
Let $\{C(\bm{x})\}_{\bm{x}\in\mathbb{Z}^d}\subset U(2d)$ be a 
family of unitary and self-adjoint square matrices of order $2d$. 
\paragraph{Coin operators}
We define a coin operator $C$ on $\mathcal{H}$ as 
a multiplication operator by $C(x)$, i.e., 
$$
	(C\Psi)(\bm{x})=C(\bm{x})\Psi(\bm{x})
	\quad
	\text{for all}\
	\bm{x}\in \mathbb{Z}^d
	\
	\text{and}
	\
	\Psi \in \mathcal{H}.
$$
By definition, $C$ is unitary on $\mathcal{H}$.
Throughout this paper, 
the following two conditions are imposed on $C$
unless otherwise specified.
\begin{itemize}
	\item ({\bf Simplicity}) 
	Each $C(\bm{x})$ has 1 as a simple eigenvalue, i.e., 
	$$
	\dim \ker (C(\bm{x})-1)=1,
	\quad \bm{x}\in\mathbb{Z}^d.
	$$
	\item ({\bf One defect}) There exist matrices 
	$C_0$ and $C_1 \in U(2d)$ such that
	\begin{align*}
		C(\bm{x})
		&=
		\begin{cases}
			C_1, & {\bm x}\in \mathbb{Z}^d\setminus\{\bm{0}\}, \\
			C_0, & {\bm x}=\bm{0}.
		\end{cases}
	\end{align*}
\end{itemize}
Because $\dim \ker(C(\bm{x})-1)=1$, 
we can take a unique normalized vector 
(up to a constant factor):
$$
\chi(\bm{x})
=\begin{pmatrix}\chi_1(\bm{x})\\
\vdots\\
\chi_d(\bm{x})\end{pmatrix}
\in \ker(C(\bm{x})-1), \quad 
\chi_j(\bm{x})=
\begin{pmatrix}
	\chi_{j,1}(\bm{x})\\
	\chi_{j,2}(\bm{x})
\end{pmatrix}
\in \mathbb{C}^2\ (j=1,\cdots,d).
$$ 
The spectral decomposition of $C(\bm{x})$
implies 
$
	C(\bm{x})=2|\chi(\bm{x})\rangle\langle\chi(\bm{x})|-1
$. 
By the one defect condition of $C$,
$\chi(\bm{x})$ is written as follows. 
\begin{align}
\label{chiPO}
		\chi(\bm{x})
		&=
		\begin{cases}
			\Phi=
			\begin{pmatrix}
				\Phi_1\\
				\vdots \\
				\Phi_d
			\end{pmatrix}\ 
			\text{with}\ 
			\Phi_j=
			\begin{pmatrix}
				\phi_{j,1}\\
				\phi_{j,2}
			\end{pmatrix}
			\in\mathbb{C}^2\ (j = 1,\cdots,d),
			& 
			\bm{x}\in\mathbb{Z}^d\setminus\{\bm{0}\}, \\
			\Omega=
			\begin{pmatrix}
				\Omega_1\\
				\vdots \\
				\Omega_d
			\end{pmatrix}\ 
			\text{with}\ 
			\Omega_j=
			\begin{pmatrix}
				\omega_{j,1}\\
				\omega_{j,2}
			\end{pmatrix}
			\in\mathbb{C}^2\ (j = 1,\cdots,d),
			& 
			\bm{x}=\bm{0}.
		\end{cases}
\end{align}
\paragraph{Evolutions and their discriminants}
Let $S$ and $C$ be as above and define an evolution operator 
$U$ on $\mathcal{H}$ as 
$$
	U=SC.
$$
$S$ and $C$ are unitary, 
and so is $U$.
We define a coisometry
$d : \mathcal{H}\to \mathcal{K}$ as 
$$
	(d\Psi)(\bm{x})=
	\langle
		\chi (\bm{x}), \Psi(\bm{x})
	\rangle_{\mathbb{C}^{2d}}\quad 
	\text{for all}\  \bm{x}\in\mathbb{Z}^d
	\ \text{and} \
	\Psi\in \mathcal{H}.
$$

\begin{lemma}\label{lem_d}
	\begin{enumerate}
	\item 
		The adjoint $d^* : \mathcal{K}\to \mathcal{H}$
		of $d$ is a multiplication operator by $\chi({\bm x})$,
		i.e.,
		\begin{align*}
			(d^{\ast}f)(\bm{x})
			&=
			\chi (\bm{x})f(\bm{x})
			\quad
			\text{for all}\ 
			\bm{x}\in\mathbb{Z}^d
			\ \text{and} \
			f\in\mathcal{K}.
		\end{align*}
	\item 
		$\displaystyle 
		d^{\ast}d=\bigoplus_{\bm{x}\in\mathbb{Z}^d}
			|\chi(\bm{x})\rangle\langle\chi(\bm{x})|$
	\ and \ $dd^{\ast}=I_{\mathcal{K}}$.
	\item $C=2d^{\ast}d-1$.
	\end{enumerate}
\end{lemma}
\begin{proof} 
(1)\ For all $f\in\mathcal{K}$, since
$
	\sum_{\bm{x}\in\mathbb{Z}^d}\|\chi(\bm{x})
	f(\bm{x})\|_{\mathbb{C}^{2d}}^2
	= \sum_{\bm{x}\in\mathbb{Z}^d}|f(\bm{x})|^2
	=\|f\|_{\mathcal{K}}^2<\infty
$, 
then the multiplication operator 
$\chi : \mathcal{K}\ni f\mapsto \chi f\in\mathcal{H}$ 
is bounded.
For all $\Psi\in\mathcal{H}$ and $f\in\mathcal{K}$, 
$
	\langle f, d\Psi\rangle_{\mathcal{K}} = \sum_{\bm{x}\in\mathbb{Z}^d}
	f(\bm{x})^{\ast}\langle \chi (\bm{x}), \Psi(\bm{x})
	\rangle_{\mathbb{C}^{2d}}
	=\langle \chi f, \Psi \rangle_{\mathcal{H}}
$. 
Thus we have $d^{\ast}f=\chi f$.\\
(2)\ For all $\Psi\in \mathcal{H}$ and $\bm{x}\in\mathbb{Z}^d$, 
$
	(d^{\ast}d\Psi)(\bm{x})
	=\chi(\bm{x})(d\Psi)(\bm{x})
	=\langle \chi(\bm{x}),\Psi(\bm{x})\rangle\chi(\bm{x})
$
holds. Then we have 
$d^{\ast}d=\bigoplus_{\bm{x}\in\mathbb{Z}^d}
|\chi(\bm{x})\rangle\langle\chi(\bm{x})|$.
On the other hand, for all $f\in\mathcal{K}$ and $\bm{x}\in\mathbb{Z}^d$, 
$
	(dd^{\ast}f)(\bm{x})=
	\langle \chi(\bm{x}), (d^{\ast}f)(\bm{x})\rangle_{\mathbb{C}^{2d}}
	=\langle \chi(\bm{x}), f(\bm{x})\chi(\bm{x})\rangle_{\mathbb{C}^{2d}}
	=f(\bm{x})
$. 
Then $dd^{\ast}=I_{\mathcal{K}}$ holds.\\
(3)\ Obviously, the result follows from 
$d^{\ast}d=\bigoplus_{\bm{x}\in\mathbb{Z}^d}
|\chi(\bm{x})\rangle\langle\chi(\bm{x})|$.
\end{proof}
Lemma \ref{lem_d} implies that
Theorem \ref{thm_SM} is applicable for the above evolution $U$.  
\medskip 
In what follows, we give an explicit form of the discriminant operator
$T$ of $U$, defined as
\[ T = d S d^*. \] 
Let $L_j$ be a shift on $\mathcal{K}$ 
by ${\bm e}_j$ ($j\in \{1,\cdots, d\}$),
i.e.,
$$
	(L_jf)(\bm{x})=f(\bm{x}+\bm{e}_j), \quad \text{for all}\
	\bm{x}\in\mathbb{Z}^d
	\ \text{and} \
	f\in\mathcal{K},
$$
by which $S_j$ can be expressed as a matrix form
$$
	S_j=
	\begin{pmatrix}
		p_j I_{\mathcal{K}} & q_jL_j \\
		q_j^{\ast}L_j^{\ast} & -p_jI_{\mathcal{K}}
	\end{pmatrix}. 
$$
We use the following notations: 
\begin{align*}
	&a_\Omega({\bm p})
	= \sum_{j=1}^d p_j 
	\langle \Omega_j, \sigma_3 \Omega_j \rangle_{\mathbb{C}^2}, 
	\quad
	a_\Phi({\bm p})
	= \sum_{j=1}^d p_j 
	\langle \Phi_j, \sigma_3 \Phi_j \rangle_{\mathbb{C}^2}, 
	\\
 	&\quad \mbox{and} \quad 
	a(\bm{p}, \bm{x})=\sum_{j=1}^dp_j
	\langle \chi_j(\bm{x}), \sigma_3 \chi_j(\bm{x}) \rangle_{\mathbb{C}^2}, 
\end{align*}
where
$\sigma_3 = \begin{pmatrix} 1 & 0 \\ 0 & -1 \end{pmatrix}$.
Observe that
\begin{equation}
\label{apx}
	a(\bm{p}, \bm{x})=a_\Omega({\bm p})\mathbbm{1}_{\{\bm{0}\}}
	(\bm{x})+ a_\Phi({\bm p})\mathbbm{1}_{\mathbb{Z}^d\setminus\{\bm{0}\}}
	(\bm{x}),
\end{equation} 
where $\mathbbm{1}_{A}$ is the characteristic function of a set $A$. 
As seen in Section 2, 
the discriminant operator $T = d Sd^*$ of $U$  
is a bounded self-adjoint on $\mathcal{K}$ and $\|T\|\leq 1$.  
\begin{lemma}\label{repT}
$T$ is expressed as
\begin{align}
\label{eq_repT}		
T=
		a(\bm{p}, \cdot)+
		\sum_{j=1}^d
		\left\{
		q_j\chi_{j,1}^{\ast}L_j\chi_{j,2} + 
		(q_j\chi_{j,1}^{\ast}L_j\chi_{j,2})^{\ast}
		\right\},
\end{align}
where $\chi_{j,1}, \chi_{j,2}$ and $a(\bm{p}, \cdot)$ denote multiplication operators.
\end{lemma}
\begin{remark}
In Section 1, we abbreviate the expression \eqref{eq_repT} as
\begin{align*}
		T=
		a(\bm{p}, \cdot)+
		\sum_{j=1}^2
		\left\{
		q_j \langle \chi_j,  L_j \sigma_+ \chi_j \rangle
		+ (q_j \langle \chi_j,  L_j \sigma_+ \chi_j \rangle)^*
		\right\}
		\ \mbox{with $\sigma_+ = \begin{pmatrix} 0 & 1 \\ 0 & 0 \end{pmatrix}$. }
	\end{align*}
\end{remark}
\begin{proof} 
For all $f\in \mathcal{K}, j\in \{1,\cdots, d\}$,
we set 
$(d^{\ast}f)_j=\chi_j f \in \ell^2(\mathbb{Z}^d ; \mathbb{C}^2)$.
Since
$d^{\ast}f=
\begin{pmatrix}
	(d^{\ast}f)_1\\
	\vdots\\
	(d^{\ast}f)_d
\end{pmatrix}$, 
we have 
$Sd^{\ast}f=
\begin{pmatrix}
	S_1(d^{\ast}f)_1\\
	\vdots \\
	S_d(d^{\ast}f)_d
\end{pmatrix}$.
By definition of $T$, the following holds 
for all $f\in \mathcal{K}$ and $\bm{x}\in \mathbb{Z}^d$: 
\begin{align*}
	(Tf)(\bm{x})&=\langle \chi(\bm{x}), (Sd^{\ast}f)(\bm{x})
	\rangle_{\mathbb{C}^{2d}}
	=\sum_{j=1}^d\langle \chi_j(\bm{x}), (S_j(d^{\ast}f)_j)(\bm{x})
	\rangle_{\mathbb{C}^{2}}\\
	&=
	\sum_{j=1}^d
	\left\langle 
		\begin{pmatrix}
			\chi_{j,1}\\
			\chi_{j,2}
		\end{pmatrix}
		(\bm{x}), 
		\left(
			\begin{pmatrix}
				p_j & q_jL_j\\
				q_j^{\ast}L_j^{\ast} & -p_j
			\end{pmatrix}
			\begin{pmatrix}
				f\chi_{j,1}\\
				f\chi_{j,2}
			\end{pmatrix}
		\right)
		(\bm{x})
	\right\rangle_{\mathbb{C}^{2}}\\
	&=
	\left(
	\left(
	a(\bm{p}, \cdot)+
		\sum_{j=1}^d
		\left\{
			q_j\chi_{j,1}^{\ast}L_j\chi_{j,2} + 
			(q_j\chi_{j,1}^{\ast}L_j\chi_{j,2})^{\ast}
		\right\}
	\right)
		f
	\right)
	(\bm{x}).
\end{align*}
\end{proof}

We close this subsection by characterizing 
the essential spectrum of the discriminant $T$. 
To this ends, we introduce a self-adjoint operator $T_0$ 
and constant $\lambda({\bm q})$ by
\begin{align*}
	T_0 
	= 
	a_{\Phi}(\bm{p})+ \sum_{j=1}^d
	(\alpha_jL_j + \alpha_j^{\ast}L_j^{\ast})
	\quad
	\mbox{and} \quad
	\lambda({\bm q}) =2\sum_{j=1}^d|\alpha_j|,
\end{align*}
where 
$\alpha_j=q_j\phi_{j,1}^{\ast}\phi_{j,2} \ (j=1,\cdots,d)$.
In Sec. 1, we set
$\displaystyle 
\lambda({\bm q}) 
	= 2\sum_{j=1}^2 \left|q_j \langle \Phi_j, \sigma_+ \Phi_j 
		\rangle_{\mathbb{C}^2}\right|$,
because  $\alpha_j=q_j \langle \Phi_j, \sigma_+ \Phi_j 
		\rangle_{\mathbb{C}^2}$.
\begin{lemma}\label{sT0}
It follows that
	\begin{align}\label{spectrumT0}
		\sigma_{\mathrm{ess}}(T)=\sigma_{\mathrm{ess}}(T_0)
		=\sigma(T_0)=[-\lambda({\bm{q}})+a_{\Phi}(\bm{p}), 
		a_{\Phi}(\bm{p})+\lambda({\bm{q}})]. 
	\end{align}
Moreover, the following conditions are equivalent:
	\begin{enumerate}
		\item $\sigma(T_0)=[-1, 1]$;
		\item $\lambda({\bm{q}})=1$; 
		\item $p_j=0$ and 
		$|\phi_{j,1}|=|\phi_{j,2}|$
		for all $j\in \{1,\cdots,d\}$. 
	\end{enumerate}
\end{lemma}
\begin{remark}
Let 
$g_\pm(\lambda) = e^{\pm i \arccos \lambda}$. 
The spectral mapping theorem (Theorem \ref{thm_SM}) 
concludes that
\[ 
\sigma_{\rm ess}(U) 
	= \{ g_-(\lambda) \mid \lambda \in \sigma(T_0) \}
		\cup 
	 \{ g_+(\lambda) \mid \lambda \in \sigma(T_0) \}. 
\]
See Figure \ref{fig01}. 
\end{remark}
\begin{proof}
Let 
${W} = T - T_0$. Then $T = T_0 + {W}$ and
\begin{align*}
	{W}
	&=
	a(\bm{p,x})-a_{\Phi}(\bm{p})
	+
	\sum_{j=1}^dq_j(\chi_{j,1}(\bm{x})^{\ast}\chi_{j,2}(\bm{x}+\bm{e}_j)
	-\phi_{j,1}^{\ast}\phi_{j,2})L_j \\
	&+
	\sum_{j=1}^dq_j^*(\chi_{j,2}(\bm{x})^{\ast}\chi_{j,1}(\bm{x}-\bm{e}_j)
	-\phi_{j,2}^{\ast}\phi_{j,1})
	L_j^*. 
\end{align*}
Because, by \eqref{chiPO} and \eqref{apx},
$({W}f)(\bm{x})=0$ for all $\bm{x}\neq \pm \bm{e}_j, \bm{0}$
and $f \in \mathcal{K}$,
\[ 
	{W} =  \beta_0 \mathbbm{1}_{\{\bm{0}\}}  
	+ \sum_{j=1}^d
	\left\{ \beta^+_j \mathbbm{1}_{\{\bm{e}_j\}}
	+  \beta^-_j  \mathbbm{1}_{\{-\bm{e}_j\}} \right\}
\]
with some constants  $\beta_0$ and $\beta^\pm_j$ ($j=1,\dots,d$). 
Because ${W}$ is compact, $\sigma_{\mathrm{ess}}(T)= \sigma_{\mathrm{ess}}(T_0)$. 

Let 
$\mathcal{F} : \mathcal{K}\to L^2([0,2\pi]^d, d\bm{k}/(2\pi)^d)$ 
be the Fourier transformation defined as the unitary extension of
$$
	(\mathcal{F}f)(\bm{k}) = 
	\hat{f}(\bm{k})=\sum_{\bm{x}\in\mathbb{Z}^d}
	e^{-i\bm{k \cdot x}}
	f(\bm{x})
	\quad \text{for all $f \in \mathcal{K}$ with finite support}.
$$
Because
$\mathcal{F}L_j\mathcal{F}^{\ast}$ 
and 
$\mathcal{F}L_j^{\ast}\mathcal{F}^{\ast}$ are
multiplication operators by $e^{ik_j}$ and $e^{-ik_j}$,
the Fourier transform $\mathcal{F}T_0\mathcal{F}^{\ast}$
of $T_0$
is also a multiplication operator by 
\begin{align}\label{range}
	\hat{T_0}({\bm k}) 
	=a_{\Phi}(\bm{p})+2\sum_{j=1}^d|\alpha_j|\cos (k_j+\theta_j), 
\end{align}
where each $\theta_j\in [0,2\pi)$ is an argument of $\alpha_j$, i.e., 
$\alpha_j=|\alpha_j|e^{i\theta_j}$
(if $\alpha_j=0$, we define $\theta_j=0$).
Because 
$\hat T_0([0,2\pi]^d) = [-\lambda_{\bm{q}}+a_{\Phi}(\bm{p}), a_{\Phi}(\bm{p})+\lambda_{\bm{q}}]$, 
we have \eqref{spectrumT0}. 

(\ref{spectrumT0}) implies that 
$a_{\Phi}(\bm{p})=0$ if $\lambda({\bm q})=1$ 
and hence that
$\sigma(T_0)=[-1, 1]$ if and only if
$
		\lambda({\bm q})=1
$.
On the other hand, 
$|q_j|\leq 1$ and the inequality of 
arithmetic and geometric means yield the inequality
$$
	\lambda({\bm q}) \leq 2\sum_{j=1}^d|\phi_{j,1}\phi_{j,2}|
	\leq \sum_{j=1}^d(|\phi_{j,1}|^2+|\phi_{j,2}|^2)=1
$$
with equality if and only if $|q_j|=1$ and $|\phi_{j,1}|=|\phi_{j,2}|$
for all $j\in \{1,\cdots,d\}$. 
This completes the proof. 
\end{proof}

\subsection{Main results}
In what follows, we prove 
the existence of discrete eigenvalues of $U_{\rm I}$.
To this end, we impose the following on the coin operator,
which corresponds to (1) in Theorem 
\ref{thm_d=2} for the case of $d=2$. 
Let
	$\sigma_1=\begin{pmatrix}
	0&1\\
	1&0
	\end{pmatrix}$. 
{We use \, $\cdot$ \, to denote the scalar product,
i.e., $\Psi \cdot \Phi = \psi_1 \phi_1 + \psi_2 \phi_2$
for $\Psi = {}^t (\psi_1, \psi_2)$, 
$\Phi = {}^t (\phi_1, \phi_2) \in \mathbb{C}^2$. }
\begin{assumption}\label{ratio}
\begin{itemize}
\item[(a)] 
{$\Phi_j 
		\cdot (\sigma_1 \Omega_{j}) =0$ }
	for all $j\in\{1,\cdots,d\}$; 
\item[(b)] 
{$\langle \Phi_l, 
			\sigma_+ \Omega_l \rangle_{\mathbb{C}^2}
		\neq 0$}
	with some $l \in\{1,\cdots, d\}$. 
\end{itemize}
\end{assumption}
Let $l$ be as in Assumption \ref{ratio}
and set
\begin{align*}
	D_{l} = \{(\bm{p, q}) \in D
	\colon  p_l q_{l}\neq 0\}.
\end{align*}
Lemma \ref{sT0} shows that
if $(\bm{p, q}) \in D_l$, then
$\sigma_{\mathrm{ess}}(T)=\sigma(T_0)\subsetneq [-1,1]$. 
Hence, 
there can exist discrete eigenvalues of $T$ in 
$[-1,1]\setminus \sigma(T_0)\neq \emptyset$. 
In order to find the discrete eigenvalue, 
we introduce a function $\mathfrak{f}: [-1,1]\setminus 
	\sigma(T_0) \to \mathbb{R}$ as follows. 
Let 
\[
	\varphi_{\bm{q}} = 
	\sum_{j=1}^d\left(
	q_j\omega_{j,2}\phi_{j,1}^{\ast}\mathbbm{1}_{\{-\bm{e}_j\}}
	+
	q_j^{\ast}\omega_{j,1}\phi_{j,2}^{\ast}
	\mathbbm{1}_{\{\bm{e}_j\}}
	\right)
	\in\mathcal{K}.
\]
For $\lambda\in [-1,1]\setminus\sigma(T_0)\neq \emptyset$,
we define
\[ 
	\mathfrak{f}(\lambda)
	=
	\lambda -a_{\Omega}(\bm{p})+
	\langle \varphi_{\bm{q}}, \psi_{\lambda} \rangle_{\mathcal{K}},
\]
where 
\begin{equation}
\label{eq_pl}
\psi_{\lambda} 
	= (T_0-\lambda)^{-1}\varphi_{\bm{q}}\in\mathcal{K}.
\end{equation} 
{Let $\sigma_- = \sigma_+^*$. 
Because $\varphi_{\bm q}$ is written as 
\begin{equation}
\label{phiq}
\varphi_{\bm q}
	= \sum_{j=1}^d  
	\left( q_j  \langle \Phi_j, \sigma_+ \Omega_j \rangle
		\mathbbm{1}_{\{-\bm{e}_j\}}
		+ q_j^{\ast}\langle \Phi_j, \sigma_- \Omega_j \rangle
			  \mathbbm{1}_{\{\bm{e}_j\}}
	\right),
\end{equation}
$({\bm p}, {\bm q}) \in D_l$ ensures that
$\varphi_{\bm q} \not \equiv 0$ and $\psi_\lambda \not\equiv 0$.
}

The next theorem plays an important role to show the eigenvalue of $T$. 
\begin{theorem}\label{f0}
Suppose that Assumption \ref{ratio} holds
and $({\bm p}, {\bm q}) \in D_l$. 
If $\mathfrak{f}$ has a zero $\lambda_\star \in [-1,1]\setminus 
	(\sigma(T_0)\cup \{a_{\Omega}(\bm{p})\})$,
then  $\lambda_\star$ is a discrete eigenvalue of $T$. 
\end{theorem}
\begin{remark}
By Lemma \ref{evaluationf} (3),
$a_{\Omega}({\bm p})$ can not be a zero of $\mathfrak{f}$
even if $a_\Omega({\bm p}) \in [-1,1]\setminus \sigma(T_0)$.
{Hence, 
$\mathfrak{f}$ has a zero $\lambda_\star \in [-1,1]\setminus 
	(\sigma(T_0)\cup \{a_{\Omega}(\bm{p})\})$
(if it exists) and 
Theorem \ref{f0} concludes that 
$\lambda_\star \in \sigma_{\rm p}(U)$. 
The SMT and Theorem \ref{f0} imply that 
$g_\pm(\lambda_\star) \in \sigma_{\rm p}(U)$
are discrete eigenvalues of $U$. 
}
See Figure \ref{fig01}. 
\end{remark}
The proof of Theorem \ref{f0} is based on 
the Feshbach projection method \cite{F, BFS}. 
This reduces the spectral analysis of $T$ to that of the Feshbach map 
$F(T, P, \lambda)$,
which is an operator defined by $T$, 
a projection $P$ suitably chosen, and a spectral parameter $\lambda$. 
Let 
$\Pi =|\mathbbm{1}_{\{\bm{0}\}}\rangle\langle \mathbbm{1}_{\{\bm{0}\}}|$
be the projection onto the subspace 
$\{ \alpha \mathbbm{1}_{\{\bm{0}\}} 
	\mid \alpha \in \mathbb{C}\} \subset \mathcal{K}$
and $\Pi^{\perp} = I_{\mathcal{K}}- \Pi$.  
Here we chose $P = \Pi^\perp$
as the projection defining the Feshbach map
and set $F(\lambda) = F(T, \Pi^\perp, \lambda)$.  
See Sec. 4 for the precise definition of $F(\lambda)$ and 
propositions used in the following proof. 
\begin{proof}[Proof of Theorem \ref{f0}]
By Proposition \ref{Fh},
$F(\lambda)$ is written as
$$
		F(\lambda) = \Pi^{\perp}\left(
		T_0-\lambda -\frac{1}{a_{\Omega}(\bm{p})-\lambda}|\varphi_{\bm{q}} \rangle 
	\langle \varphi_{\bm{q}} |
	\right)
	\Pi^{\perp},
	\quad \lambda \in \mathbb{C} \setminus \{a_\Omega({\bm p})\}.
	$$ 
Let $\lambda_\star \in [-1,1]\setminus (\sigma(T_0)\cup\{a_{\Omega}(\bm{p})\})$ be a zero of $\mathfrak{f}$,
i.e., $\mathfrak{f}(\lambda_\star)=0$,
and let $\psi_{\lambda_\star}$ be defined in \eqref{eq_pl}
with $\lambda = \lambda_\star$.  
Because by Proposition \ref{psilambda}, 
$\psi_{\lambda_\star}\in\Ran\Pi^{\perp}\setminus \{0\}$,
\begin{align*}
F(\lambda_\star) \psi_{\lambda_\star}
& =	\Pi^\perp \left(
		T_0-\lambda_\star -\frac{1}{a_{\Omega}(\bm{p})-\lambda_\star}|\varphi_{\bm{q}} \rangle 
	\langle \varphi_{\bm{q}} |
	\right)
\psi_{\lambda_\star} \\
& = 
	\left(
		1-\frac{\langle \varphi_{\bm{q}}, \psi_{\lambda_\star} 
		\rangle}{a_{\Omega}(\bm{p})-\lambda}
	\right)
	\varphi_{\bm{q}} =
	- \frac{\mathfrak{f}(\lambda_\star)}{a_{\Omega}(\bm{p})-\lambda_\star} 
	\varphi_{\bm{q}} =0.
\end{align*}
This completes the proof,
because by Proposition \ref{feshbach},
$\lambda_\star \in \sigma_{\rm p}(T)$
is equivalent that
$\ker F(\lambda_\star)$ is non trivial,
which is confirmed by Proposition \ref{psilambda} again. 
\end{proof}
The following is a criterion for $\mathfrak{f}$ to have a zero. 
\begin{theorem}
\label{mainth}
Suppose that Assumption \ref{ratio} holds and 
$(\bm{p}, \bm{q})\in D_l$. 
{%
\begin{itemize}
\item[(1)] 
$\mathfrak{f}$ has a zero $\lambda_\star 
\in \mathbb{T}_- := [-1, -\lambda({\bm q})+a_{\Phi}({\bm p}))$
if 
\begin{equation}
\lambda({\bm{q}})(\lambda({\bm{q}})+a_{\Omega}(\bm{p})
		-a_{\Phi}(\bm{p})) 
		<
		\|\varphi_{\bm{q}} \|^2 
		\leq
		(1+a_{\Omega}(\bm{p}))
		\frac{(1+a_{\Phi}(\bm{p}))^2-\lambda({\bm{q}})^2}{1+a_{\Phi}(\bm{p})};
		\label{up}
\end{equation}
\item[(2)] 
$\mathfrak{f}$ has a zero $\lambda_\star 
\in \mathbb{T}_+ = (\lambda({\bm q})+a_{\Phi}({\bm p}), 1]$
if 
\begin{equation}
\lambda({\bm{q}})(\lambda({\bm{q}})-a_{\Omega}(\bm{p})
		+a_{\Phi}(\bm{p})) 
		<
		\|\varphi_{\bm{q}} \|^2 
		\leq
		(1-a_{\Omega}(\bm{p}))
		\frac{(1-a_{\Phi}(\bm{p}))^2-\lambda({\bm{q}})^2}{1-a_{\Phi}(\bm{p})}. 
		\label{down}
\end{equation}
\end{itemize}
}
\end{theorem}
Thanks to Lemma \ref{neqpm1} below,
the right-hand sides of (\ref{up}) and (\ref{down}) make sense.  
The proof of 
Theorem \ref{mainth} will be stated in the last section. 
\begin{lemma}\label{neqpm1}
Suppose that Assumption \ref{ratio} holds and $(\bm{p}, \bm{q})\in D_l$.
Then, 
\[ a_{\Phi}(\bm{p})\neq \pm 1
	\quad \text{and} \quad  a_{\Omega}(\bm{p})\neq \pm 1.
\]
\end{lemma}
\begin{proof}
Suppose $a_{\Phi}(\bm{p})=-1$. 
By the definition of $a_{\Phi}(\bm{p})$ and  $\|\Phi\|^2=1$, 
\begin{equation*}
	-1
	=
	\sum_{j\notin A}p_j(|\phi_{j,1}|^2-|\phi_{j,2}|^2)
	\quad \text{and} \quad
	1
	=
	\sum_{j\notin A}(|\phi_{j,1}|^2+|\phi_{j,2}|^2),
\end{equation*}
where $A = \{ j\in \{1,\cdots, d\} \mid \phi_{j,1} =\phi_{j,2}= 0\}$. 
Summing the above two equations, 
we get
$
	0=\sum_{j\notin A}
	\left\{(1+p_j)|\phi_{j,1}|^2+(1-p_j)|\phi_{j,2}|^2\right\},
$
which, combined with $1\pm p_j \geq 0$, 
implies that for $j \not\in A$,
\begin{equation}
\label{eq383}
(1-p_j)|\phi_{j,2}|^2=0. 
\end{equation}
By Assumption \ref{ratio} and $({\bm p}, {\bm q}) \in D_l$,
$\phi_{l,2} \not=0$, $p_l \not= 1$, and hence $l \not\in A$. 
This contradicts \eqref{eq383}.  
Therefore $a_{\Phi}(\bm{p})\neq -1$.
The remainder can be shown similarly.
\end{proof}
To state our main result, we introduce the following assumption. 
\begin{assumption}\label{ap0}
$a_{\Omega}(\bm{p}_0)\neq a_{\Phi}(\bm{p}_0)$ holds
with some $\bm{p}_0 \in \{-1,1\}^d$. 
\end{assumption}

\begin{remark}
If $d=1$, then Assumptions \ref{ratio}-\ref{ap0} are not compatible. 
See \cite{FFSloc} for $d=1$. 
\end{remark}

\begin{theorem}[Existence of eigenvalues]\label{MThm}
	Let $d\geq 2$ and 
	suppose that Assumptions \ref{ratio}-\ref{ap0} holds. 
	Then, 
	there exists $\delta >0$ such that
	if $(\bm{p},\bm{q}) \in D_l$ satisfies 
$\|(\bm{p},\bm{q})-(\bm{p}_0, \bm{0})\|_{\mathbb{R}^d\times\mathbb{C}^d}<\delta$, 
then there exist eigenvalues of $U$. 
Moreover, the following hold. 
\begin{itemize}
\item[(1)]
If, in addition, $a_\Omega({\bm p}_0) < a_\Phi({\bm p}_0)$,
then $g_-(\lambda_\star)$ and $g_+(\lambda_\star)$
are eigenvalues of $U_{\rm I}$ with some 
$\lambda_\star \in \mathbb{T}_-$;
\item[(2)]
If, in addition, $a_\Omega({\bm p}_0) > a_\Phi({\bm p}_0)$,
then $g_-(\lambda_\star)$ and $g_+(\lambda_\star)$
are eigenvalues of $U_{\rm I}$ with some 
$\lambda_\star \in \mathbb{T}_+$. 
\end{itemize}
\end{theorem}
\begin{proof}
Observe that 
$\varphi_{\bm{q}}\neq 0$ and $\lambda({\bm{q}})>0$
whenever $(\bm{p}, {\bm q}) \in D_l$. 
By continuity, we have
\begin{align}\label{4th}
\lim_{{\bm q} \to {\bm 0}} \|\varphi({\bm{q}})\| =   
\lim_{{\bm q} \to {\bm 0}}  \lambda({\bm{q}}) = 0, 
\
\lim_{{\bm p} \to {\bm p}_0} a_{\Omega}(\bm{p})
	=  a_{\Omega}(\bm{p}_0), 
\ \mbox{and} \
\lim_{{\bm p} \to {\bm p}_0} a_{\Phi}(\bm{p})
	=  a_{\Phi}(\bm{p}_0). 
\end{align}
Suppose that $a_{\Omega}(\bm{p}_0)<a_{\Phi}(\bm{p}_0)$. 
Then, 
there exists $\delta_0 > 0$ such that 
if $(\bm{p}, \bm{q})\in D_l$ and $\|(\bm{p},\bm{q})-(\bm{p}_0, \bm{0})\|_{\mathbb{R}^d\times\mathbb{C}^d}<\delta_0$, then $\lambda({\bm{q}})(\lambda({\bm{q}})
	+a_{\Omega}(\bm{p})-a_{\Phi}(\bm{p}))<0$. 
By Lemma \ref{neqpm1} and (\ref{4th}), 
$$
	\lim_{(\bm{p},\bm{q})\to(\bm{p}_0, \bm{0})}
	(1+a_{\Omega}(\bm{p}))
	\frac{(1+a_{\Phi}(\bm{p}))^2-\lambda({\bm{q}})^2}{1+a_{\Phi}(\bm{p})}
	=(1+a_{\Omega}(\bm{p}_0))(1+a_{\Phi}(\bm{p}_0)) > 0.
$$
Hence (\ref{up}) holds if $({\bm p}, {\bm q}) \in D_l$
satisfies $\|(\bm{p},\bm{q})-(\bm{p}_0, \bm{0})\|_{\mathbb{R}^d\times\mathbb{C}^d}<\delta$
with some $\delta > 0$. 
Similarly, 
$a_{\Omega}(\bm{p}_0)>a_{\Phi}(\bm{p}_0)$
concludes that (\ref{down}) holds.
Applying Theorems \ref{thm_SM} and \ref{mainth} completes the proof.
\end{proof}
\begin{remark}
Theorem \ref{MThm} 
has demonstrated the existence of eigenvalues of $U_{\rm I}$
for sufficiently small ${\bm q}$. 
It would be interesting to study
the existence of eigenvalues of $U_{\rm I}$ without such a condition.  
\end{remark}

\begin{example}
Let $d=2$ and set 
$$
	\Phi:=\frac{1}{\sqrt{2}}
	\begin{pmatrix}
	1\\
	1\\
	0\\
	0
	\end{pmatrix}, \quad 
	\Omega:=\frac{1}{2}
	\begin{pmatrix}
	1\\
	-1\\
	\sqrt{2}\\
	0
	\end{pmatrix}, \quad
	\bm{p}_0=(1,1).
$$
Then,
$
	\frac{1}{2}=a_{\Omega}(\bm{p}_0)
	> a_{\Phi}(\bm{p}_0) = 0
$
and all assumptions in Theorem \ref{MThm} are satisfied
with $l=1$. 
Hence $U$ has two eigenvalues 
if $(\bm{p, q}) \in D_1$ and
$\|(\bm{p, q})  - (\bm{p}_0, \bm{0})\|$ is sufficiently small. 
More precisely, $g_\pm(\lambda_\star) \in \sigma_{\rm p}(U_{\rm I})$
with some $\lambda_\star \in \mathbb{T}_+$
if ${\bm p}$ satisfies  
\begin{align}\label{excon}
p_2<\frac{5}{2}-\frac{1}{2p_1^2}, \quad 
1<p_1^2+\frac{4}{9}p_2^2.
\end{align}
{%
This is because, in this case, \eqref{excon} is equivalent 
to \eqref{down} in Theorem \ref{mainth}.   
}%
\end{example}


\section{Feshbach map}
\subsection{Definition of the Feshbach map}
In this subsection, we define
the Feshbach map of the discriminant operator $T$. 
Recall that $\Pi = |\mathbbm{1}_{\{\bm{0}\}}\rangle\langle \mathbbm{1}_{\{\bm{0}\}}|$ and 
	$\Pi^{\perp} = I_{\mathcal{K}}- \Pi$. 
	Let $\lambda \in \mathbb{C}$ and 
	$(\Pi (T-\lambda)\Pi)_{\Ran\Pi}$ be a following operator on 
	${\Ran}\Pi$: 
	$$
		(\Pi (T-\lambda)\Pi)_{{\Ran}\Pi} : 
		\Ran\Pi \ni f \mapsto (\Pi (T-\lambda)\Pi) f \in 
		\Ran\Pi.
	$$

\begin{lemma}\label{neqlam}
	The following conditions are equivalent: 
	\begin{enumerate}
		\item 
			$\lambda \neq a_{\Omega}(\bm{p})$;
		\item 
			There exists an inverse operator of 
			$(\Pi (T-\lambda)\Pi)_{{\Ran}\Pi}$.
	\end{enumerate}
	In this case, 
	\begin{align}\label{ranpi}
	(\Pi (T-\lambda)\Pi)_{{\Ran}\Pi}^{-1}
		=\frac{1}{a_{\Omega}(\bm{p})-\lambda}I_{{\Ran}\Pi},
	\end{align}
	where $I_{{\Ran}\Pi}$ is an identity map on ${\Ran}\Pi$.
\end{lemma}
\begin{proof}
Simple calculation show that
$\langle \mathbbm{1}_{\{\bm{0}\}}, 
\chi_{j,1}^{\ast}L_j\chi_{j,2}\mathbbm{1}_{\{\bm{0}\}}\rangle =0$
for all $j\in \{1,\cdots,d\}$
and
$
	\Pi (T-\lambda)\Pi = (a_{\Omega}(\bm{p})-\lambda)\Pi
$ for all $\lambda\in\mathbb{C}$. 
Hence,
$$
	(\Pi (T-\lambda)\Pi)_{{\Ran}\Pi}=(a_{\Omega}(\bm{p})-\lambda)
	I_{{\Ran}\Pi}\quad
	\text{for all}\ 
	\lambda\in\mathbb{C}.
$$
Therefore, (1) and (2) are equivalent and (\ref{ranpi}) holds 
for $\lambda \neq a_{\Omega}(\bm{p})$.
\end{proof}
Lemma \ref{neqlam} guarantees that
the operator $$
	F(\lambda) = 
	\Pi^{\perp}(T-\lambda)\Pi^{\perp}-\Pi^{\perp}T\Pi 
	(\Pi (T-\lambda)\Pi)_{{\Ran}\Pi}^{-1} 
	\Pi T \Pi^{\perp}
$$
is well-defined 
whenever $\lambda\in\mathbb{C}\setminus\{a_{\Omega}(\bm{p})\}$. 
$F(\lambda)$ is called the Feshbach map of $T$. 
The following proposition 
reveals an isospectral property of the Feshbach map.  

\begin{proposition}\label{feshbach}
Let $\lambda\in\mathbb{C}\setminus\{a_{\Omega}(\bm{p})\}$. 
Then, the following are equivalent:
\begin{itemize}
\item[(1)] $\lambda \in \sigma_{\rm p}(T)$;
\item[(2)] $\ker F(\lambda)$ is non trivial. 	
\end{itemize}
In this case, $\dim \ker (T-\lambda) = \dim \ker F(\lambda)$. 
\end{proposition}
\begin{proof}
See \cite{BFS}.
\end{proof}

\begin{proposition}\label{Fh}
Let $\lambda \in\mathbb{C}\setminus\{ a_{\Omega}(\bm{p})\}$.
Then, $F(\lambda)$ is written as 
	$$
		F(\lambda) = \Pi^{\perp}\left(
		T_0-\lambda -\frac{1}{a_{\Omega}(\bm{p})-\lambda}|\varphi_{\bm{q}} \rangle 
	\langle \varphi_{\bm{q}} |
	\right)
	\Pi^{\perp}.
	$$
\end{proposition}
\begin{proof}
A simple calculation yields
$\Pi T \Pi^{\perp} 
= |\mathbbm{1}_{\{\bm{0}\}}\rangle \langle \varphi_{\bm{q}} |$.
By definition,
\begin{align*}
	F(\lambda) 
& =
	\Pi^{\perp}
	\left\{
		T-\lambda-\frac{1}{a_{\Omega}(\bm{p})-\lambda}
		(\Pi T\Pi^{\perp})^{\ast}(\Pi T\Pi^{\perp})
	\right\}
	\Pi^{\perp} \\
& =
	\Pi^{\perp}
	\left(
		T-\lambda -\frac{1}{a_{\Omega}(\bm{p})-\lambda}
		|\varphi_{\bm{q}}\rangle \langle\varphi_{\bm{q}} |
	\right)
	\Pi^{\perp}.
\end{align*}
It suffices to show 
$\Pi^{\perp}T\Pi^{\perp}=\Pi^{\perp}T_0\Pi^{\perp}$.
By Lemma \ref{repT}, 
\begin{align}\label{perpT}
	\Pi^{\perp}T\Pi^{\perp}=
	\Pi^{\perp}a(\bm{p},\cdot)\Pi^{\perp}+
	\sum_{j=1}^d
	\left\{
		q_j\Pi^{\perp}\chi_{j,1}^{\ast}L_j\chi_{j,2}\Pi^{\perp} + 
		(q_j\Pi^{\perp}\chi_{j,1}^{\ast}L_j\chi_{j,2}\Pi^{\perp})^{\ast}
	\right\}.
\end{align}
The first term of the right-hand side of (\ref{perpT}) is calculated as
\begin{align*}
	\Pi^{\perp}a(\bm{p},\cdot)\Pi^{\perp}
	&=
	\sum_{j=1}^dp_j\Pi^{\perp}(|\chi_{j,1}|^2-
	|\chi_{j,2}|^2)\Pi^{\perp}\\
	&=
	\sum_{j=1}^dp_j
	\sum_{\bm{x}\neq \bm{0}}
	\sum_{\bm{y}\neq \bm{0}}
	|\mathbbm{1}_{\{\bm{x}\}}\rangle\langle\mathbbm{1}_{\{\bm{x}\}}, 
	(|\chi_{j,1}(\bm{x})|^2-
	|\chi_{j,2}(\bm{x})|^2)
	\mathbbm{1}_{\{\bm{y}\}}\rangle
	\langle \mathbbm{1}_{\{\bm{y}\}}|\\
	&=
	\sum_{j=1}^dp_j
	\sum_{\bm{x}\neq \bm{0}}
	\sum_{\bm{y}\neq \bm{0}}
	|\mathbbm{1}_{\{\bm{x}\}}\rangle\langle\mathbbm{1}_{\{\bm{x}\}}, 
	(|\phi_{j,1}|^2-
	|\phi_{j,2}|^2)
	\mathbbm{1}_{\{\bm{y}\}}\rangle
	\langle \mathbbm{1}_{\{\bm{y}\}}|\\
	&=
	\Pi^{\perp}a_{\Phi}(\bm{p})\Pi^{\perp}.
\end{align*}
On the other hand, 
\begin{align*}
	\Pi^{\perp}\chi_{j,1}^{\ast}L_j\chi_{j,2}\Pi^{\perp}
	&=
	\sum_{\bm{x}\neq \bm{0}}
	\sum_{\bm{y}\neq \bm{0}}
	|\mathbbm{1}_{\{\bm{x}\}}\rangle\langle\mathbbm{1}_{\{\bm{x}\}}, 
	\chi_{j,1}^{\ast}L_j\chi_{j,2}
	\mathbbm{1}_{\{\bm{y}\}}\rangle
	\langle \mathbbm{1}_{\{\bm{y}\}}|\\
	&=
	\sum_{\bm{x}\neq \bm{0}}
	\sum_{\bm{y}\neq \bm{0}}
	(\chi_{j,1}^{\ast}L_j\chi_{j,2}
	\mathbbm{1}_{\{\bm{y}\}})(\bm{x})
	|\mathbbm{1}_{\{\bm{x}\}}\rangle
	\langle \mathbbm{1}_{\{\bm{y}\}}|\\
	&=
	\sum_{\bm{x}\neq \bm{0}}
	\sum_{\bm{y}\neq \bm{0}}
	\chi_{j,1}^{\ast}(\bm{x})
	\chi_{j,2}(\bm{x}+\bm{e}_j)
	\mathbbm{1}_{\{\bm{y}\}}(\bm{x}+\bm{e}_j)
	|\mathbbm{1}_{\{\bm{x}\}}\rangle
	\langle \mathbbm{1}_{\{\bm{y}\}}|\\	
	&=
	\sum_{\bm{x}\neq \bm{0}}
	\phi_{j,1}^{\ast}
	\phi_{j,2}
	|\mathbbm{1}_{\{\bm{x}\}}\rangle
	\langle \mathbbm{1}_{\{\bm{x}+\bm{e}_j\}}|\\
	&=
	\Pi^{\perp}\phi_{j,1}^{\ast}
	\phi_{j,2}L_j\Pi^{\perp}.
\end{align*}
Hence, $\Pi^{\perp}T\Pi^{\perp}=\Pi^{\perp}T_0\Pi^{\perp}$.
This completes the proof. 
\end{proof}

\begin{remark}
In the proof of Proposition \ref{Fh}, one defect condition plays an essential role. 
If the coin has two or more defect, then we can not conclude 
$\Pi^{\perp}\chi_{j,1}^{\ast}L_j\chi_{j,2}\Pi^{\perp}
=\Pi^{\perp}\phi_{j,1}^{\ast}
	\phi_{j,2}L_j\Pi^{\perp}$.
\end{remark}
\subsection{Non-triviality of the kernel of $F(\lambda)$}
The following proposition ensures the non-triviality of $\ker F(\lambda)$
in the proof of Theorem \ref{f0}. 
Recall that $\psi_{\lambda} 
	= (T_0-\lambda)^{-1}\varphi_{\bm{q}}\in\mathcal{K}$.
\begin{proposition}\label{psilambda}
Suppose that Assumption \ref{ratio} holds
and $({\bm p}, {\bm q}) \in D_l$. 
Then, 
$\psi_{\lambda}\in\Ran\Pi^{\perp}\setminus\{0\}$ for all 
	$\lambda\in [-1,1]\setminus \sigma(T_0)$.	
\end{proposition}


To prove Proposition \ref{psilambda},
we need the following lemma. 
\begin{lemma}\label{psiranpi}
	The following are equivalent:
	\begin{enumerate}
	\item 
		$\psi_{\lambda}\in \Ran\Pi^{\perp}$;
	\item 
		$\displaystyle{\int_{[0,2\pi]^d}\hat{\psi_{\lambda}}(\bm{k})
		\frac{d\bm{k}}{(2\pi)^d}}=0$.
	\end{enumerate}
\end{lemma}
\begin{proof}
Since $\Pi^{\perp}\psi_{\lambda}=\psi_{\lambda}$ is equivalent 
that 
$\mathcal{F}\Pi^{\perp}\mathcal{F}^{\ast}\hat{\psi_{\lambda}}
=\hat{\psi_{\lambda}}$, 
then (1) is equivalent the following:
\begin{align}\label{ranF}
	\hat{\psi_{\lambda}}\in \Ran \mathcal{F}\Pi^{\perp} \mathcal{F}^{\ast}.
\end{align}
Because by direct calculation, 
$\mathcal{F}\mathbbm{1}_{\{\bm{0}\}}=\mathbbm{1}_{[0,2\pi]^d}$,
the following holds:
\begin{align}
	\mathcal{F}\Pi^{\perp}\mathcal{F}^{\ast}\hat{\psi_{\lambda}}
	&=
	\mathcal{F}(I_{\mathcal{K}}-\Pi)\mathcal{F}^{\ast}\hat{\psi_{\lambda}}\notag\\
	&=
	(I_{\mathcal{FK}}-|\mathbbm{1}_{[0,2\pi]^d}\rangle
	\langle\mathbbm{1}_{[0,2\pi]^d}|)\hat{\psi_{\lambda}}\notag\\
	&=
	\hat{\psi_{\lambda}}-\left(
	\int_{[0,2\pi]^d}\hat{\psi_{\lambda}}(\bm{k})\frac{d\bm{k}}{(2\pi)^d}
	\right)
	\mathbbm{1}_{[0,2\pi]^d}.\label{psi0}
\end{align}
By (\ref{ranF}) and (\ref{psi0}), 
(1) holds if and only if
$\displaystyle 
	\left(
		\int_{[0,2\pi]^d}\hat{\psi_{\lambda}}(\bm{k})\frac{d\bm{k}}{(2\pi)^d}
	\right)
\mathbbm{1}_{[0,2\pi]^d}=0$. 
This proves the lemma. 
\end{proof}

\begin{proof}[Proof of Proposition \ref{psilambda}]
Fix $\lambda \in [-1,1] \setminus \sigma(T_0)$. 
Because $\lambda \not\in \sigma(T_0)$
and  $\varphi_{\bm q}\not=0$,
we observe that $\psi_\lambda \not=0$. 
By using (\ref{range}) and changing variables,
we have
\begin{align}
	&
	\int_{[0,2\pi]^d}\hat{\psi_{\lambda}}(\bm{k})\frac{d\bm{k}}{(2\pi)^d}
	=
	\int_{[0,2\pi]^d}\frac{\hat{\varphi_{\bm{q}}}(\bm{k})}{\hat{T_0}(\bm{k})-\lambda}
	\frac{d\bm{k}}{(2\pi)^d}\notag\\
	&=
	\int_{[0,2\pi]^d}
	\frac{\sum_{j=1}^d\left(\varphi_{\bm{q}}(\bm{e}_j)e^{-ik_j}
	+\varphi_{\bm{q}}(-\bm{e}_j)e^{ik_j}\right)}
	{a_{\Phi}(\bm{p})+2\sum_{j=1}^d|\alpha_j|\cos(k_j+\theta_j)-\lambda}
	\frac{d\bm{k}}{(2\pi)^d}\notag\\
	&=
	\int_{\prod_{j=1}^d[\theta_j,2\pi+\theta_j]}
	\frac{\sum_{j=1}^d\left(\varphi_{\bm{q}}(\bm{e}_j)e^{-i(t_j-\theta_j)}
	+\varphi_{\bm{q}}(-\bm{e}_j)e^{i(t_j-\theta_j)}\right)}
	{a_{\Phi}(\bm{p})+2\sum_{j=1}^d|\alpha_j|\cos t_j-\lambda}
	\frac{d\bm{t}}{(2\pi)^d}\notag\\
	&=
	\int_{[0,2\pi]^d}
	\frac{\sum_{j=1}^d\left(\varphi_{\bm{q}}(\bm{e}_j)e^{i\theta_j}
	+\varphi_{\bm{q}}(-\bm{e}_j)e^{-i\theta_j}\right)\cos t_j}
	{a_{\Phi}(\bm{p})+2\sum_{j=1}^d|\alpha_j|\cos t_j-\lambda}
	\frac{d\bm{t}}{(2\pi)^d}
	\label{intsum}.
\end{align}
{By \eqref{phiq}, 
we observe that
\begin{align}\label{keisuu}
	\varphi_{\bm{q}}(\bm{e}_j)e^{i\theta_j}
	+\varphi_{\bm{q}}(-\bm{e}_j)e^{-i\theta_j}
= q_j^* e^{i \theta_j} 
	\langle \Phi_j, \sigma_- \Omega_j \rangle,
	+ q_j e^{-i \theta_j} 
	\langle \Phi_j, 
		\sigma_+ \Omega_j \rangle.
\end{align}
Let $B = \{ j \mid q_j \not= 0, 
	\ \langle \Phi_j, \sigma_+ \Omega_j \rangle \not= 0\}$.
Assumption \ref{ratio} (b) and $({\bm p}, {\bm q}) \in D_l$
imply $B \not=\emptyset$. 
If $j \not\in B$, then the right-hand side (RHS) of \eqref{keisuu} is zero,
because by $\sigma_1 = \sigma_+ +  \sigma_-$, 
Assumption \ref{ratio} (a) implies that
\begin{equation} 
\label{eq4}
|\langle \Phi_j, \sigma_+ \Omega_j \rangle|
	= |\Phi_j \cdot (\sigma_+ \Omega_j)|
	= |\Phi_j \cdot (\sigma_- \Omega_j)|
	= |\langle \Phi_j, \sigma_- \Omega_j \rangle|. 
\end{equation}  
Let $j \in B$. By \eqref{eq4},
we have $\phi_{j,1}\not=0$, $\phi_{j,2} \not=0$, and hence
$\alpha_j=q_j\phi_{j,1}^{\ast}\phi_{j,2}\not=0$.
Using $e^{i\theta_j}=q_j \phi_{j,1}^*\phi_{j,2}/|\alpha_j|$,
we observe that
\begin{align*}
\mbox{RHS of \eqref{keisuu}}
& = \left\langle \Phi_j, 
	\begin{pmatrix} 0 &   q_j e^{-i \theta_j}  \\
	 q_j^* e^{i \theta_j}  & 0
	\end{pmatrix} \Omega_j \right\rangle  
= \frac{|q_j|^2}{|\alpha_j|} 
 	\left\langle 
 	\begin{pmatrix} 0 &   \phi_{j,1}\phi_{j,2}^*  \\
	 \phi_{j,1}^*\phi_{j,2} & 0
	\end{pmatrix}
	\Phi_j,  \Omega_j \right\rangle  \\
& = \frac{|q_j|^2 \phi_{j,1}^*\phi_{j,2}^*}{|\alpha_j|} 
		\, \Phi_j \cdot (\sigma_1 \Omega_j)
	= 0.
\end{align*}
Therefore, 
\[ \varphi_{\bm{q}}(\bm{e}_j)e^{i\theta_j}
	+\varphi_{\bm{q}}(-\bm{e}_j)e^{-i\theta_j} = 0,
	\quad j=1,\dots,d, \]
which, in conjunction with \eqref{intsum},
gives
}
$\displaystyle 
\int_{[0,2\pi]^d}\hat{\psi_{\lambda}}(\bm{k})
	\frac{d\bm{k}}{(2\pi)^d}=0$.  
Lemma \ref{psiranpi} concludes the proof.
\end{proof}

\section{
Zeros of $\mathfrak{f}$
}

%



In this section we prove Theorem \ref{mainth}. 
We henceforth suppose that Assumption \ref{ratio} holds 
and fix $(\bm{p, q})\in D_l$. 
Recall that 
$\mathfrak{f}: [-1,1]\setminus\sigma(T_0)\to \mathbb{R}$
is defined as
\[ 
	\mathfrak{f}(\lambda)
	=
	\lambda -a_{\Omega}(\bm{p})+
	\langle \varphi_{\bm{q}}, \psi_{\lambda} \rangle_{\mathcal{K}}
\]
and that we set
$\mathbb{T}_- = [-1, -\lambda({\bm q})+a_{\Phi}({\bm p}))$
and $\mathbb{T}_+ = (\lambda({\bm q})+a_{\Phi}({\bm p}), 1]$. 
We need the following lemmas. 

\begin{lemma}\label{monotone}
The function $\mathfrak{f}$ is continuously differentiable 
and monotonically increasing.
\end{lemma}
\begin{proof}
The lemma is evident from 
$$
\mathfrak{f}^\prime(\lambda) = 
	1+ \int_{[0,2\pi]^d}\frac{|\hat{\varphi_{\bm{q}}}(\bm{k})|^2}
	{(\hat{T_0}(\bm{k})-\lambda)^2}\frac{d\bm{k}}{(2\pi)^d}
	>0.
$$
\end{proof}



\begin{lemma}\label{evaluationf}
	The following hold:
	\begin{enumerate}
	\item If $\lambda\in 
	\mathbb{T}_-
	$, then
	\begin{align}\label{Llambda}
		\lambda-a_{\Omega}(\bm{p})
		+\frac{\|\varphi_{\bm{q}}\|^2}{a_{\Phi}(\bm{p})-\lambda}
		<
		\mathfrak{f}(\lambda)
		<
		\lambda-a_{\Omega}(\bm{p})+\frac{a_{\Phi}(\bm{p})-\lambda}
		{(a_{\Phi}(\bm{p})-\lambda)^2-\lambda({\bm{q}})^2}\|\varphi_{\bm{q}}\|^2.
	\end{align}
	\item If $\lambda\in 
	\mathbb{T}_+
	$, then 
	\begin{align}
		\lambda-a_{\Omega}(\bm{p})
		+\frac{\|\varphi_{\bm{q}}\|^2}{a_{\Phi}(\bm{p})-\lambda}
		>
		\mathfrak{f}(\lambda)
		>
		\lambda-a_{\Omega}(\bm{p})+\frac{a_{\Phi}(\bm{p})-\lambda}
		{(a_{\Phi}(\bm{p})-\lambda)^2-\lambda({\bm{q}})^2}\|\varphi_{\bm{q}}\|^2.
	\end{align}
	\item 
If $a_{\Omega}(\bm{p})\in [-1,1]\setminus \sigma(T_0)$, then
	$\mathfrak{f}(a_{\Omega}(\bm{p}))\neq 0$.
	\end{enumerate}
\end{lemma}
\begin{proof}
Let $\lambda\in
\mathbb{T}_-
$. 
Since $\varphi_{\bm{q}} \neq 0$, 
$\mathfrak{f}(\lambda)$ can be written as
$$
	\mathfrak{f}(\lambda) = \lambda -a_{\Omega}(\bm{p}) + 
	\|\varphi_{\bm{q}}\|^2
	\int_{\sigma(T_0)}g_{\lambda}(x)d
	\langle \varphi_{\bm{q}}, E_{T_0}(x)\varphi_{\bm{q}} 
	\rangle /\|\varphi_{\bm{q}}\|^2, 
$$
where $E_{T_0}(\cdot)$ is the spectral measure of $T_0$ and 
$g_{\lambda}(x)=\frac{1}{x-\lambda}$.
Note that $\langle \varphi_{\bm{q}}, E_{T_0}(\cdot)\varphi_{\bm{q}} 
\rangle /\|\varphi_{\bm{q}}\|^2$ is a probability measure on 
$\sigma(T_0)=[a_{\Phi}(\bm{p})-\lambda_{\bm{q}}, a_{\Phi}(\bm{p})+\lambda_{\bm{q}}]$.
By Jensen's inequality, we have 
\begin{align*}
	\int_{\sigma(T_0)}g_{\lambda}(x)d
	\langle \varphi_{\bm{q}}, E_{T_0}(x)\varphi_{\bm{q}} 
	\rangle /\|\varphi_{\bm{q}}\|^2
	&>
	g_{\lambda}\left(
	\int_{\sigma(T_0)}xd\langle \varphi_{\bm{q}}, E_{T_0}(x)\varphi_{\bm{q}} 
	\rangle /\|\varphi_{\bm{q}}\|^2
	\right)
	\\
	&= 
	g_{\lambda}
	\left(\langle 
	\varphi_{\bm{q}}, T_0\varphi_{\bm{q}}\rangle /\|\varphi_{\bm{q}}\|^2
	\right).
	\end{align*}
Because
$\langle \varphi_{\bm{q}}, T_0\varphi_{\bm{q}}\rangle 
/\|\varphi_{\bm{q}}\|^2= a_{\Phi}(\bm{p})$,
we have, 
\begin{align}\label{leftf}
	\mathfrak{f}(\lambda )>\lambda -a_{\Omega}(\bm{p}) + 
	\|\varphi_{\bm{q}}\|^2g_{\lambda}(a_{\Phi}(\bm{p}))
	=
	\lambda-a_{\Omega}(\bm{p})+\frac{\|\varphi_{\bm{q}}\|^2}{a_{\Phi}(\bm{p})-\lambda}.
\end{align}
Let $u : [a_{\Phi}(\bm{p})-\lambda({\bm{q}}), a_{\Phi}(\bm{p})+\lambda({\bm{q}})]
\to \mathbb{R}$ be a linear function such that
$u(a_{\Phi}(\bm{p})-\lambda({\bm{q}}))=g_{\lambda}(a_{\Phi}(\bm{p})-\lambda({\bm{q}}))$ and 
$u(a_{\Phi}(\bm{p})+\lambda({\bm{q}}))=g_{\lambda}(a_{\Phi}(\bm{p})+\lambda({\bm{q}}))$, 
i.e., 
$$
	u(x)=\frac{-x+2a_{\Phi}(\bm{p})-\lambda}{(a_{\Phi}(\bm{p})-\lambda)^2
	-\lambda({\bm{q}})^2}.
$$
By the convexity of $g_{\lambda}$, we have
\begin{align*}
	\int_{\sigma(T_0)}g_{\lambda}(x)d
	\langle \varphi_{\bm{q}}, E_{T_0}(x)\varphi_{\bm{q}} 
	\rangle /\|\varphi_{\bm{q}}\|^2
	&<
	\int_{\sigma(T_0)}u(x)d
	\langle \varphi_{\bm{q}}, E_{T_0}(x)\varphi_{\bm{q}} 
	\rangle /\|\varphi_{\bm{q}}\|^2 \\
	&=
	\frac{-\langle \varphi_{\bm{q}}, T_0\varphi_{\bm{q}} 
	\rangle /\|\varphi_{\bm{q}}\|^2 +2a_{\Phi}(\bm{p})
	-\lambda}{(a_{\Phi}(\bm{p})-\lambda)^2-\lambda({\bm{q}})^2}\\
	&=
	\frac{a_{\Phi}(\bm{p})-\lambda}{(a_{\Phi}(\bm{p})-\lambda)^2
	-\lambda({\bm{q}})^2}.
\end{align*}
Hence, 
\begin{align}\label{rightf}
	\mathfrak{f}(\lambda)
	<
	\lambda-a_{\Omega}(\bm{p})+\frac{a_{\Phi}(\bm{p})-\lambda}
	{(a_{\Phi}(\bm{p})-\lambda)^2-\lambda_{\bm{q}}^2}\|\varphi_{\bm{q}}\|^2.
\end{align}
(\ref{leftf}) and (\ref{rightf}) imply (\ref{Llambda}).
Hence (1) is proved. 
The same proof works for (2). 

We prove (3). 	
If $a_{\Omega}(\bm{p})
\in \mathbb{T}_-
$,
then, $a_\Omega({\bm p}) < a_\Phi({\bm p})$. 
By (\ref{leftf}), we have $\mathfrak{f}(a_{\Omega}(\bm{p}))>0$. 
Similarly,  if $a_{\Omega}(\bm{p})\in 
\mathbb{T}_+
$, then
$\mathfrak{f}(a_{\Omega}(\bm{p}))<0$.
\end{proof}
	
Because by Lemma \ref{monotone} 
$\mathfrak{f}$ is monotonically increasing , 
$\mathfrak{f}(\lambda)$ has a zero in $\mathbb{T}_-$ 
{if and only if}
\begin{align}
f(-1)\leq 0 
\quad\text{and}\quad
		\lim_{\lambda \,\uparrow \, 
		a_{\Phi}(\bm{p})-\lambda({\bm{q}})}
		\mathfrak{f}(\lambda) >0.
	\tag{L}
\end{align}
Similarly,
$\mathfrak{f}(\lambda)$
has a zero in $\mathbb{T}_+$ 
{if and only if}
\begin{align}
f(1)\geq 0
\quad\text{and}\quad
		\lim_{\lambda \, \downarrow \, 
		a_{\Phi}(\bm{p})+\lambda({\bm{q}})}
		\mathfrak{f}(\lambda)<0.
	\tag{R}
\end{align}
\begin{proof}[Proof of Theorem \ref{mainth}]
Let $\lambda \in \mathbb{T}_-$. 
By Lemma \ref{evaluationf}, 
(L) holds if 
\begin{align}\label{twocon}
	\begin{cases}\displaystyle
		-1-a_{\Omega}(\bm{p})+
		\frac{a_{\Phi}(\bm{p})+1}{(a_{\Phi}(\bm{p})+1)^2-\lambda({\bm{q}})^2}
		\|\varphi_{\bm{q}}\|^2\leq 0, \\
		\displaystyle
		0 < 
		a_{\Phi}(\bm{p})-\lambda({\bm{q}})-a_{\Omega}(\bm{p})+
		\frac{\|\varphi_{\bm{q}}\|^2}{\lambda({\bm{q}})},
	\end{cases}
\end{align}
which is equivalent to (\ref{up}).
Thus,
(\ref{up}) concludes (L). 
This proves (1) of Theorem \ref{mainth}. 
The same proof works for (2). 
\end{proof}


%
%


\section*{Acknowledgments}
This work was supported by JSPS Grant in Aid for Young Scientists (B) 26800054.

\end{document}